\newif\ifshort
\shortfalse

\documentclass[a4paper,UKenglish,cleveref,autoref,thm-restate]{lipics-v2021}

\pdfoutput=1 
\hideLIPIcs  

\title{Navigating the Complexity Landscape of Nominee Selection in Schulze Voting}
\author{Katarína Cechlárová}
    {P.J. \v Saf\'arik University, Slovakia}
    {katarina.cechlarova@upjs.sk}
    {https://orcid.org/0000-0002-9641-1351}
    {Supported by VEGA 1/0585/24 and APVV-21-0369.}

\author{J\"org Rothe}
    {Heinrich-Heine-Universität Düsseldorf, Germany}
    {rothe@hhu.de}
    {https://orcid.org/0000-0002-0589-3616}
    {Supported in part by Deutsche Forschungsgemeinschaft under DFG research grant \linebreak[4]RO\nobreakdash-1202/21\nobreakdash-2 (project 438204498).}

\author{\v{S}imon Schierreich}
    {AGH University of Krakow, Poland \and
    Czech Technical University in Prague, Czechia}
    {schiesim@fit.cvut.cz}
    {https://orcid.org/0000-0001-8901-1942}
    {\flag{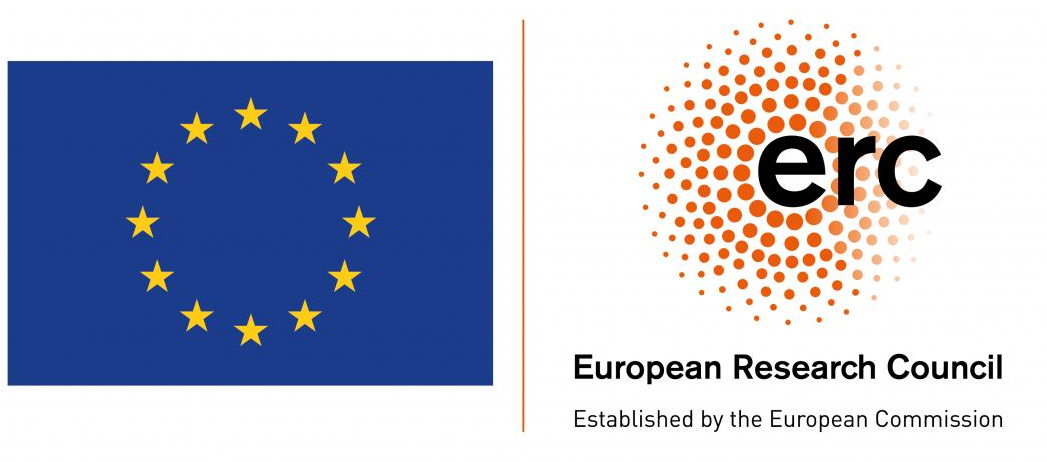}Supported from the European Research Council (ERC) under the European Union's Horizon 2020 research and innovation programme (grant agreement No 101002854) and by the European Union co-funded project Robotics and Advanced Industrial Production (reg. no. CZ.02.01.01/00/22\_008/0004590).}

\author{Ildikó Schlotter}
    {ELTE Centre for Economic and Regional Studies, Hungary \and
    Budapest University of Technology and Economics, Hungary}{schlotter.ildiko@krtk.elte.hu}
    {https://orcid.org/0000-0002-0114-8280}
    {Supported by the Hungarian Academy of Sciences under its Momentum Programme (LP2021-2) and its J\'anos Bolyai Research Scholarship.}

\authorrunning{K. Cechlárová, J. Rothe, Š. Schierreich, and I. Schlotter} 

\Copyright{Katarína Cechlárová, Jörg Rothe, Šimon Schierreich, and Ildikó Schlotter} 



\category{} 

\relatedversion{} 

\acknowledgements{The authors are grateful to Jana Woitaschik for her useful feedback on preliminary versions of the paper.}

\nolinenumbers 

\EventEditors{John Q. Open and Joan R. Access}
\EventNoEds{2}
\EventLongTitle{42nd Conference on Very Important Topics (CVIT 2016)}
\EventShortTitle{CVIT 2016}
\EventAcronym{CVIT}
\EventYear{2016}
\EventDate{December 24--27, 2016}
\EventLocation{Little Whinging, United Kingdom}
\EventLogo{}
\SeriesVolume{42}
\ArticleNo{23}

\usepackage{amssymb,amsmath,verbatim}
\usepackage{booktabs}

\usepackage{xspace}

\usepackage{hyperref}
\usepackage{cleveref}
\usepackage{doi}

\usepackage{algorithm}
\usepackage[noend]{algpseudocode}

\usepackage[dvipsnames]{xcolor}
\usepackage{graphicx}
\usepackage{tikz}
\usetikzlibrary{backgrounds,calc,shapes.geometric,arrows.meta,decorations.markings,patterns.meta,calc}
\tikzset{
myarc/.style={
 	-{Stealth[length=3mm,width=1.8mm]},line width=1pt
	}
}

\newcommand{\probName}[1]{\textsc{#1}\xspace}
\def\SchPP{\probName{Schulze Possible President}}
\def\SchNP{\probName{Schulze Necessary President}}
\newcommand{\BalancedSat}{\probName{$2$-Balanced $3$-SAT}}

\newcommand{\voters}{V}
\newcommand{\numVoters}{{|V|}}
\newcommand{\candidates}{C}

\newcommand{\nomination}{N}
\newcommand{\election}{\mathcal{E}}
\newcommand{\electionI}{\election=(\candidates,\voters)}

\newcommand{\numParties}{k}
\newcommand{\maxPartySize}{\sigma}
\newcommand{\majG}{G(\election)}
\newcommand{\majGN}{G(\election(\nomination))}
\DeclareMathOperator{\weight}{\omega}
\def\beats{\mathsf{str}}
\def\inneigh{\delta^-}
\def\outneigh{\delta^+}

\newcommand{\colorSize}{{x}}
\newcommand{\edgeSetSize}{{y}}

\def\clauses{\Psi} 
\def\clause{\psi} 
\def\pathc{q}
\def\distc{p}
\def\Prfl{\Pi}
\def\true{\texttt{true}}
\def\false{\texttt{false}}

\newcommand{\N}{\mathbb{N}}

\def\P{{\cal P}}

\newcommand{\ol}[1]{\overline{#1}}
\newcommand{\ola}[1]{\overleftarrow{#1}}
\newcommand{\ora}[1]{\overrightarrow{#1}}
\newcommand{\wt}[1]{\widetilde{#1}}

\newcommand{\Oh}[1]{{\mathcal{O}\left(#1\right)}}

\NewDocumentCommand{\cc}{ O{} O{} m }{\mbox{%
    \expandafter\ifx\expandafter\relax\detokenize{#2}\relax\else{#2-}\fi%
    \textrm{#3}%
    \expandafter\ifx\expandafter\relax\detokenize{#1}\relax\else{-#1}\fi%
    }\xspace}
\newcommand{\DeclareComplexityLowerBound}[1]{%
  \expandafter\newcommand\csname #1\endcsname{\cc{#1}}%
  \expandafter\newcommand\csname #1h\endcsname{\cc[hard]{#1}}%
  \expandafter\newcommand\csname #1hness\endcsname{\cc[hardness]{#1}}%
  \expandafter\newcommand\csname #1c\endcsname{\cc[complete]{#1}}%
  \expandafter\newcommand\csname #1cness\endcsname{\cc[completeness]{#1}}%
}
\newcommand{\DeclareComplexityLowerBoundParam}[2]{%
  \expandafter\newcommand\csname #1\endcsname[1][#2]{\cc{#1[##1]}}%
  \expandafter\newcommand\csname #1h\endcsname[1][#2]{\cc[hard]{#1[##1]}}%
  \expandafter\newcommand\csname #1hness\endcsname[1][#2]{\cc[hardness]{#1[##1]}}%
  \expandafter\newcommand\csname #1c\endcsname[1][#2]{\cc[complete]{#1[##1]}}%
  \expandafter\newcommand\csname #1cness\endcsname[1][#2]{\cc[completeness]{#1[##1]}}%
}
\newcommand{\DeclareComplexityUpperBound}[1]{
  \expandafter\newcommand\csname #1\endcsname{\cc{#1}}%
}

\DeclareComplexityLowerBound{NP}
\DeclareComplexityLowerBound{paraNP}
\DeclareComplexityLowerBound{coNP}
\DeclareComplexityLowerBoundParam{W}{1}
\DeclareComplexityLowerBoundParam{coW}{1}
\DeclareComplexityUpperBound{FPT}
\DeclareComplexityUpperBound{XP}

\def\mypropsymbol{\diamond}
\def\propchain{(\hyperlink{prop:chain}{$\mypropsymbol$})}

\Crefname{claim}{Claim}{Claims}

\definecolor{myYellow}{RGB}{252,199,18} 
\definecolor{myBlue}{RGB}{51,102,204}
\definecolor{myRed}{RGB}{179,13,26}
\definecolor{myGreen}{RGB}{38,140,56}
\definecolor{myGray}{RGB}{130,128,133} 

\usepackage[title]{appendix}

\newcommand{\linkproof}[1]{%
\ifshort
    $\star$%
\else
    \hyperref[#1]{$\star$}%
\fi
}

\begin{document}
\maketitle              
\begin{abstract}
    We study the \probName{Possible President} problem and the \probName{Necessary President} problem for Schulze voting, a rule that, due to its many desirable axiomatic properties, is popular in practice. In both problems, we are given an election with the candidates partitioned into a set of \emph{parties}, and we are interested in questions about a given \emph{distinguished party}. In the \probName{Possible President} problem, we ask whether it is possible for the parties to each nominate exactly one candidate such that the nominee of the distinguished party is a Schulze winner of the resulting election with only the nominees running. In the \probName{Necessary President} problem, we ask whether the distinguished party's nominee is a Schulze winner of the resulting election, irrespective of the nomination from the other parties. Rothe and Woitaschik~\cite{rot-woi:c:possible-and-necessary-president-problem-in-schulze-voting} have shown that \probName{Possible President} is \NPc and \probName{Necessary President} is \coNP-complete for Schulze elections. We complement and improve their results by a more fine-grained analysis: we determine the parameterized complexity of both problems with respect to all possible parameterizations, where we consider each of three natural parameters---the number of voters, the maximum party size, and the number of parties---to be either a constant, a parameter, or unbounded. In particular, we obtain dichotomies regarding the number of voters for both problems.
\end{abstract}

\section{Introduction}\label{sec:introduction}

One of the core topics in computational social choice (COMSOC) is the study of strategic behavior in voting.
Initially, the main focus of interest was on manipulation~\cite{bar-tov-tri:j:manipulating,con-san-lan:j:when-hard-to-manipulate}, electoral control~\cite{bar-tov-tri:j:control,hem-hem-rot:j:destructive-control}, and bribery~\cite{fal-hem-hem:j:bribery,fal-hem-hem-rot:j:llull-copeland-full-techreport}, as surveyed in the book chapters by Conitzer and Walsh~\cite{con-wal:b:handbook-comsoc-manipulation}, Faliszewski and Rothe~\cite{fal-rot:b:handbook-comsoc-control-and-bribery}, and Baumeister and Rothe~\cite{bau-rot:b-2nd-edition:economics-and-computation-preference-aggregation-by-voting}.
Meanwhile, the focus of attention in COMSOC has shifted to other areas.
 %
Among the topics increasingly gaining attention is \emph{candidate nomination} in settings with multiple parties, each facing the task of selecting a nominee for an upcoming election.

In the \probName{Possible President} problem, we are given an election and a partition of the candidates into parties, and we ask whether each party can nominate exactly one candidate such that the nominee of a distinguished party  wins the resulting election with only the nominees running; if so, this candidate is a \emph{possible president}.
In the closely related \probName{Necessary President} problem, the question on the same input is whether the distinguished party has a \emph{necessary president}, i.e., a candidate who is guaranteed to win the resulting election regardless of whom the other parties nominate.
The \probName{Possible} and \probName{Necessary President} problems were introduced by Faliszewski \emph{et al.}~\cite{fal-gou-lan-les-mon:c:how-hard-is-it-for-a-party-to-nominate-an-election-winner} who studied their computational complexity for plurality voting.

These problems, which arguably are especially relevant in the real world,\footnote{In essentially all countries, political elections are held by parties nominating their top candidates as their representatives.}
have also been studied by
Misra~\cite{mis:c:parameterized-party-nominations},
Cechl{\'{a}}rov{\'{a}} \emph{et al.}~\cite{cec-les-tre-han-han:j:hardness-of-candidate-nomination},
Schlotter \emph{et al.}~\cite{sch-cec:c:candidate-nomination-for-condorcet-consistent-voting-rules,sch-cec-tre:j:parameterized-complexity-of-candidate-nomination-for-elections-based-on-positional-scoring-rules}
Faliszewski \emph{et al.}~\cite{fal-kaz-lis-sch-tur:c:computing-equilibrium-nominations}
Cechlárová and Schlotter~\cite{cec-sch:c:necessary-president},
and Rothe and Woitaschik~\cite{rot-woi:c:possible-and-necessary-president-problem-in-schulze-voting}.
In particular, Cechl{\'{a}}rov{\'{a}} \emph{et al.}~\cite{cec-les-tre-han-han:j:hardness-of-candidate-nomination} have studied them for several voting rules, including scoring rules (such as Borda, $k$-approval, and $k$-veto) and Condorcet-consistent rules (such as Copeland and maximin voting, a.k.a.\ the Simpson--Kramer rule) in terms of their computational complexity.
In addition, they formulated integer programs to solve the \probName{Possible President} problem, and performed experiments on both real-world and synthetic data.
Complementing the classical complexity analysis, Misra~\cite{mis:c:parameterized-party-nominations} and Schlotter \emph{et al.}~\cite{sch-cec:c:candidate-nomination-for-condorcet-consistent-voting-rules,sch-cec-tre:j:parameterized-complexity-of-candidate-nomination-for-elections-based-on-positional-scoring-rules} have studied this problem in terms of its parameterized complexity.
Here, we continue the work of Rothe and Woi\-ta\-schik~\cite{rot-woi:c:possible-and-necessary-president-problem-in-schulze-voting} who have shown \NPcness of the \probName{Possible President} problem and \coNPcness of the \probName{Necessary President} problem for Schulze elections; both their results hold even if no party contains more than three members.

Schulze's voting rule~\cite{sch:j:schulze-voting} is quite popular and has been widely used in practice, mainly due to its many desirable axiomatic properties: Among other properties, Schulze voting is monotonic, clone-independent, reversal-symmetric, and Condorcet-consistent.
This is why it has been applied for decision-making by the Pirate Parties of several countries, including Sweden and Australia, as well as by the Wikimedia Foundation, Kubernetes, and the Debian Vote Engine.
Schulze voting has been thoroughly investigated in COMSOC, especially regarding strategic behavior such as control attacks, for example by Parkes and Xia~\cite{par-xia:c:strategic-schulze-ranked-pairs}, Menton and Singh~\cite{men-sin:c:control-complexity-schulze}, and more recently by Maushagen \emph{et al.}~\cite{mau-nic-nue-rot-see:c:toward-completing-the-picture-of-control-in-schulze-and-ranked-pairs-elections}.


\paragraph*{Our Contribution}
We obtain a dichotomy for both the \probName{Possible} and the \probName{Necessary President} problem regarding the number of voters: For two voters, both problems are solvable in polynomial (actually, even in linear) time, whereas for three or more voters, the \probName{Possible President} problem is \NPc and the \probName{Necessary President} problem \coNPc.
In fact, these results hold even if each party has at most two members, with the exception of \probName{Necessary President} for an odd number of voters, where we were able to show \coNPcness for parties with at most three candidates.
Our dichotomies reveal that both of our problems are \paraNPh for the combination of two parameters---the number of voters and the maximum number of candidates in a party.
The third natural parameter we study is the number of parties, and we establish \Whness regarding this parameter for both problems, even for a constant number of voters.
We thus determine the parameterized complexity for all possible parameter combinations, that is, for every choice of viewing the number of voters, the maximum party size, and the number of parties as (a)~a constant, (b)~a parameter, or (c)~unbounded, we have pinpointed the parameterized complexity of both the \probName{Possible} and the \probName{Necessary President} problem for the resulting parameterization.

Although restricting the number of voters to be a constant does not guarantee fixed-parameter tractability with respect to the number of parties, we show that for the special case of three voters, \probName{Possible President} becomes fixed-parameter tractable when parameterized by the number of parties.
%


\section{Preliminaries}\label{sec:preliminaries}

Let us give some basic background from social choice, graph, and complexity theory.
In particular, we define Schulze voting and the \probName{Possible} and \probName{Necessary President} problems.


\subsection{Schulze Voting}

In social choice theory, an \emph{election} is a pair~$\electionI$, where $C$ is a set of candidates and~$\voters$ a list of votes over~$\candidates$, often expressed as a linear order $\succ_v$ over~$\candidates$ for each $v \in \voters$.
When specifying a vote~$v$'s preferences in an election, we usually omit the symbol~$\succ_v$.
For example, if~$C = \{a, b, c\}$ is a candidate set and~$b$ is preferred to~$a$ and~$a$ to~$c$ in~$v$, we write~$v: b\ a\ c$ instead of~$v: b\succ_v a\succ_v c$.
For any subset~$D \subseteq \candidates$, when $D$ occurs in some vote~$v$, this means that the members of~$D$ occur in some arbitrary, fixed order in~$v$; e.g., if~$D = \{a, b\}$, then the vote~$c\ D$ is a shorthand for~$c\ a\ b$.
Further, again assuming some fixed order of the candidates,~$\ora{D}$ gives the members of~$D$ in that order and $\ola{D}$ orders them backwards; e.g., $c\ \ora{D}$ is a shorthand for~$c\ a\ b$ and $c\ \ola{D}$ a shorthand for~$c\ b\ a$.

We define the \emph{(weighted) majority graph~$\majG=(\candidates,H)$ of~$\election$} as an (edge-weighted) directed graph (digraph) over~$\candidates$ that contains an arc~$(a,b)$ if and only if more votes in~$\voters$ prefer~$a$ to~$b$ than~$b$ to~$a$, that is,~$|\{v \in \voters : a\succ_v b\}|>|\{v \in V: b \succ_v a\}|$, and the difference $|\{u\in V: a\succ_v b\}|-|\{v \in V: b\succ_v a\}|$ is the weight~$\weight(a,b)$ of $(a,b)$ in~$\majG$.
The subgraph of~$\majG$  containing only arcs with weight exactly~$k$ will be denoted by~$G_k(\election)$.
The \emph{strength of a path~$P$ in~$\majG$} is the minimum weight of all arcs along~$P$.
The \emph{beatpath strength of~$a$ over~$b$} is the maximum strength of any path leading from~$a$ to~$b$ in~$\majG$, and is denoted by~$\beats(a,b)$;
if there is no path from~$a$ to~$b$ in~$\majG$, then we set~$\beats(a,b)=0$.
We say that a candidate~$a$ is a \emph{Schulze winner of~$(\candidates,\voters)$} if for every other candidate~$b \in \candidates \setminus\{a\}$, it holds that~$\beats(a,b) \geq \beats(b,a)$.

We conclude with an easy observation about the weighted majority graph $\majG$, which we use to simplify our proofs.

\begin{observation}
    Let $\electionI$ be an election and $\election'$ be an election constructed from $\election$ by adding two voters $v$ and $\hat{v}$ such that the ballot of $v$ is $\ola{C}$ and the ballot of $\hat{v}$ is $\ora{C}$. Then $\majG = G(\election')$.
\end{observation}

A consequence of the above observation is that whenever we show a hardness result for an instance with $c$ candidates, the same hardness immediately generalize to $c + 2i$ candidates for any $i\in\N$. We will not stress this explicitly in the remainder of the paper and our proofs.


\subsection{Possible and Necessary President}

Faliszewski \emph{et al.}~\cite{fal-gou-lan-les-mon:c:how-hard-is-it-for-a-party-to-nominate-an-election-winner} introduced the model of \emph{elections with nominations}, which has been studied by several researchers since then~\cite{mis:c:parameterized-party-nominations,cec-les-tre-han-han:j:hardness-of-candidate-nomination,sch-cec-tre:j:parameterized-complexity-of-candidate-nomination-for-elections-based-on-positional-scoring-rules,sch-cec:c:candidate-nomination-for-condorcet-consistent-voting-rules,fal-kaz-lis-sch-tur:c:computing-equilibrium-nominations,rot-woi:c:possible-and-necessary-president-problem-in-schulze-voting}.
In this model, we represent an \emph{election with parties} by a triple $(C,V,\mathcal{P})$, where $(C,V)$ is an election as defined above and $\mathcal{P} = \{P_1, \ldots, P_k\}$ is a partition of $C$ into $k$ parties, i.e., $\bigcup_{i=1}^{k} P_{i} = C$ and $P_i \cap P_j = \emptyset$ for all $i, j \in [n]$ with $i \neq j$, where $[n]$ denotes $\{1, \ldots, n\}$.

In an election with parties,~$(C, V, \mathcal{P})$, each party nominates exactly one candidate among its members, i.e., a \emph{(valid) nomination}~$N \subseteq C$ requires that~$|P_i \cap N| = 1$ for all~$P_i \in \mathcal{P}$.
We then consider the corresponding \emph{reduced election}
$(N,V)$, where all votes in $V$ are restricted to contain only the nominees from~$N$. Every party is represented by its nominee in $(N,V)$, i.e., the winning parties are those whose nominees win the reduced election. Moreover, the weighted majority graph of the reduced election $(N,V)$ induced by~$\nomination$ will be denoted by~$\majGN$, and the beatpath strength of~$a$ over~$b$ is in this graph will be denoted by~$\beats_N(a,b)$.

Like Faliszewski \emph{et al.}~\cite{fal-gou-lan-les-mon:c:how-hard-is-it-for-a-party-to-nominate-an-election-winner}, we here adopt the \emph{nonunique-winner model} where for a nominee to win it is enough to be one winner, possibly among several others.
Note that Cechl{\'{a}}rov{\'{a}} \emph{et al.}~\cite{cec-les-tre-han-han:j:hardness-of-candidate-nomination} and Schlotter \emph{et al.}~\cite{sch-cec-tre:j:parameterized-complexity-of-candidate-nomination-for-elections-based-on-positional-scoring-rules} adopted the \emph{unique-winner model} that requires a unique winner of the reduced election.

We will study two
problems, defined specifically for Schulze voting: In \SchPP, given an election $(C, V, \mathcal{P})$ with parties and a distinguished party $P_i \in P$, we ask whether there exists a valid nomination $N \subseteq C$ such that $P_i$'s nominee is a Schulze winner of the reduced election.

In \SchNP, we are given the same input, and the question now is whether $P_i$ can nominate a candidate that is a Schulze winner of each reduced election resulting from any valid nomination of the other parties.

Depending on which problem we consider, the distinguished party's nominee is called the \emph{possible president} and the \emph{necessary president}, respectively.


\subsection{Directed Graphs}

In this section we summarize the graph-theoretic notation we will use.

Let~$G=(U,E)$ be a digraph.
Given a vertex~$u \in U$, we denote by~$\inneigh_G(u)=\{v:(v,u) \in E\}$ the set of its \emph{in-neighbors} and by~$\outneigh_G(u)=\{v:(u,v) \in E\}$ the set of its \emph{out-neighbors}.
We say that vertex~$v$ is \emph{reachable} from vertex~$u$ in~$G$ if there is a directed path from~$u$ to~$v$.
A digraph is \emph{strongly connected} if every vertex is reachable from every other vertex. A \emph{strongly connected component} (SCC) of~$G$ is a maximal strongly connected subgraph of~$G$.

A digraph~$G=(C,H)$ is \emph{acyclic} if it contains no directed cycle.
A digraph~$G=(C,H)$ is \emph{transitive} if for any three vertices~$u,v,z \in V$, we have the following: If~$(u,v)\in E$ and~$(v,z)\in E$, then~$(u,z)\in E$ as well.
If a transitive digraph contains no cycles of length two (i.e., an arc~$(u,v)$ and its opposite~$(v,u)$), then it is acyclic.


\subsection{Classical and Parameterized Complexity}

We assume the reader to be familiar with the foundations of classical and parameterized complexity theory.
Regarding the former, we will consider the complexity classes \NP (``nondeterministic polynomial time'') and \coNP (the class of complements of \NP sets) and \NPhness and \NPcness as well as \coNPhness and \coNPcness with respect to the \emph{polynomial-time many-one reducibility}, referring to the standard textbooks~\cite{gar-joh:b:int,aro-bar:b:complexity,rot:b:cryptocomplexity} for more background and formal definitions.
Regarding the latter, we will consider the parameterized complexity classes FPT (``fixed-parameter tractability'') and \W (the parameterized analogue of \NP) and \Whness with respect to the \emph{polynomial-time parameterized reducibility}, again referring to the standard textbook~\cite{cyg-fom-kow-lok-mar-pil-pil-sau:b:parameterized-algorithms} for more background and formal definitions.


\section{Two Voters}

We study two problems:
We shall ask whether a candidate~$\distc \in C$ is a
winner of some reduced election (i.e., a possible president), or if $\distc$ is a winner of all reduced elections where $\distc$ has been nominated (i.e., a necessary president).
For simplicity, we assume that the party containing candidate~$\distc$ is a singleton.

%

If~$\election$ is an election with two voters, then it suffices to consider the unweighted version of the majority graph~$\majG$, as each of its arcs has weight equal to~$2$.
Furthermore, we make the following observation.

\begin{observation}\label{obs:Schulze:twoVoters:transitiveMajorityGraph}
    Let~$\electionI$ be an election with~$\numVoters=2$. Then~$\majG$ is transitive.
\end{observation}
\begin{proof}
  For two voters,~$(u,v)$ is an arc in~$\majG$ if both voters prefer candidate~$u$ to candidate~$v$.
  Hence, if~$(u,v)$ and~$(v,z)$ are both present in~$\majG$, then both voters prefer~$u$ to~$v$ and, simultaneously, both of them prefer~$v$ to~$z$.
  Consequently, both voters prefer~$u$ to~$z$, and the assertion follows.
\end{proof}

Next, we prove a necessary and satisfying condition for a candidate to be a Schulze winner in an election with two voters.

\begin{lemma}\label{thm:Schulze:twoVoters:characterisation}
    Let~$\electionI$ be an election with~$\numVoters=2$. Then~$c \in \candidates$ is a Schulze winner of~$\election$ if and only if~$\inneigh_{\majG}(c)=\emptyset$.
\end{lemma}
\begin{proof}
    Note that~$\majG$ contains neither loops nor directed~$2$-cycles by definition.
    Thus, if~$\majG$ contains an arc~$(u,c)$ then, due to transitivity of~$\majG$, there exists no path from~$c$ to~$u$.
    Thus~$\beats(c,u)=0$ while~$\beats(u,c)=2$, and so~$c$ is not a Schulze winner.

    Conversely, assume that~$\majG$ contains no arc~$(u,c)$.
    For~$c$ not to be a Schulze winner, there must exist some candidate~$z$ such that~$\beats(z,c)>\beats(c,z)$.
    Since the beatpath strength of any candidate over another is either~$0$ or~$2$, it follows that~$\beats(z,c)=2$.
    In particular, there must exist some path from~$z$ to~$c$, contradicting our assumption.
\end{proof}

And with \Cref{obs:Schulze:twoVoters:transitiveMajorityGraph} and \Cref{thm:Schulze:twoVoters:characterisation} in hand, we can prove the main result of this section, which is a linear time algorithm for both problems we study in this work.

\begin{theorem}\label{thm:PP:Schulze:twoVoters:poly}\label{thm:NP:Schulze:twoVoters:poly}
  Given an election~$\electionI$ with~$\numVoters=2$,
  both \SchPP and \SchNP can be solved in linear time.
\end{theorem}
\begin{proof}
    To check whether~$\distc$ is a possible president in~$\election$, we simply delete all candidates~$u$ such that~$(u,\distc)$ is an arc in~$\majG$, as by \Cref{thm:Schulze:twoVoters:characterisation} such candidates cannot be nominated in a reduced election where~$\distc$ is a Schulze winner.
    If all parties remain nonempty after this deletion, $\distc$ is a possible president, as all possible nominations
    leave~$\distc$ without incoming arcs. By contrast, $\distc$ is not a possible president in~$\election$ if some party becomes empty as a result of this deletion.

    By \Cref{thm:Schulze:twoVoters:characterisation}, the
    candidate~$\distc$ from the distinguished party is a necessary president in~$\election$ if and only if~$\inneigh_{\majG}(\distc) = \emptyset$:
    First, if
    $\inneigh_{\majG}(\distc) = \emptyset$, then
    $\inneigh_{\majGN}(\distc) = \emptyset$ is true for any set~$\nomination$ of nominated candidates.
    Conversely, if~$(u,\distc) \in \inneigh_{\majG}(\distc)$, then nominating~$u$ already guarantees that~$\distc$ is \emph{not} a winner of the resulting reduced election (regardless of the remaining nominations), so~$\distc$ is not a necessary president in~$\election$.
    Hence, both problems can be solved in linear time.
\end{proof}


\section{Three or More Voters}

In our hardness proofs, we shall use (polynomial-time)
reductions from the \NPc problem \BalancedSat \cite{ber-kar-sco:tr:balanced-sat}:
Given a Boolean formula~$\varphi$ over variables~$x_1,\dots, x_n$ such that each clause contains exactly three literals, and each variable occurs in~$\varphi$ exactly twice unnegated and exactly twice negated, is $\varphi$ satisfiable?
Let us introduce some notation.
We shall denote the two occurrences of~$x_i$ and of~$\bar{x}_i$ as~$x_i^1$, $x_i^2$ and as~$\bar{x}_i^1$, $\bar{x}_i^2$, respectively.
Furthermore, if~$\varphi$ contains clauses~$\clause_1,\dots,\clause_m$, then we denote by~$\ell^1_j$, $\ell^2_j$, and $\ell^3_j$ the three literal occurrences appearing in clause~$\clause_j$ (in some arbitrary order).
For example, if~$\clause_3 = (x_1 \vee \ol{x}_2 \vee x_4)$ with the clause~$\clause_3$ containing the first occurrence of the literal~$x_1$, the first occurrence of the literal~$\ol{x}_2$, and the second occurrence of the literal~$x_4$, then~$\ell_3^1=x_1^1$, $\ell_3^2=\ol{x}_2^1$, and~$\ell_3^3=x_4^2$.


\subsection{Possible President}


We start with studying the complexity of \SchPP.

\subsubsection{Intractability for a Constant Number of Voters and Small Parties}

\begin{theorem}
\label{thm:PP:Schulze:threeVoters:NPh}
  \SchPP is \NPc, even if no party contains more than two candidates and the number
  of voters is an odd constant at least three.
\end{theorem}
\begin{proof}
    That \SchPP is in \NP is obvious~\cite{rot-woi:c:possible-and-necessary-president-problem-in-schulze-voting}: Given an instance of this problem, nondeterministically guess a valid set of nominees and for each set guessed, verify in polynomial time whether the nominee of the distinguished party is a Schulze winner of the resulting reduced election.\footnote{Recall that Schulze winners can be determined in polynomial time~\cite{sch:j:schulze-voting,sor-vas-xu:c:fine-grained-complexity-schulze}.}

    To present our \NPhness reduction, we first build a machinery with which we will construct our election in a way that its
    majority graph is of a certain desired form.
    We will construct the votes $v_1$, $v_2$, and $v_3$ of our election as follows.

    \proofsubparagraph{Building a Path as the Majority Graph.}
    To begin with, we construct an election over candidate set~$\{\pathc_1,q_2,\dots,q_t\}$ whose majority graph is the following digraph~$G_{\Prfl}$: It contains all arcs along the directed path~$Q(q_1,q_2,\dots,q_t)$, as well as all arcs pointing ``backward'' on this path
    except to the direct predecessor, i.e., all arcs~$(q_i,q_j)$
    with~$i>j+1$.
    We call such a digraph the \emph{tournament based on path~$Q$}, and we denote it by~$T_Q$.
    A preference profile~$\Prfl_Q$ yielding such a graph~$T_Q$ as its majority graph can be given as follows; for simplicity, we assume that~$t$ is odd\footnote{For even~$t$, we can obtain the appropriate profile by a taking the construction for~$t+1$ and simply deleting the candidate~$q_{t+1}$.}
    (see \Cref{fig:PP:Schulze:threeVoters:NPh:tournament_path} for an illustration):
    \begin{align}
        \notag
        v_1\colon&\quad \pathc_t\ \pathc_{t-2}\ \pathc_{t-1}\ \pathc_{t-4}\ \pathc_{t-3}\ \cdots\ \pathc_3\ \pathc_4\ \pathc_1\ \pathc_2; \\
        \label{eq:PP:Schulze:threeVoters:NPh:path_profile}
        v_2\colon&\quad \pathc_1\ \pathc_2\ \pathc_3\  \cdots\ \pathc_{t-1}\ \pathc_t; \\
        \notag
        v_3\colon&\quad \pathc_{t-1}\ \pathc_t\ \pathc_{t-3}\ \pathc_{t-2}\ \cdots\ \pathc_4\ \pathc_5\ \pathc_2\ \pathc_3\ \pathc_1.
    \end{align}

    \begin{figure}[bt!]
        \centering
        \scalebox{0.95}{
        \begin{tikzpicture}[
            every node/.style={draw,circle,inner sep=0pt, minimum width=15pt}]
            \node (q1) at (0,0) {$q_1$};
            \node (q2) at (2,0) {$q_2$};
            \node (q3) at (4,0) {$q_3$};
            \node (q4) at (6,0) {$q_4$};
            \node (q5) at (8,0) {$q_5$};
            \node (q6) at (10,0) {$q_6$};
            \node (q7) at (12,0) {$q_7$};

            \draw[myarc,myGreen,line width=1pt]    (q1) to (q2);
            \draw[myarc,myRed,line width=1pt] (q2) to (q3);
            \draw[myarc,myGreen,line width=1pt]    (q3) to (q4);
            \draw[myarc,myRed,line width=1pt] (q4) to (q5);
            \draw[myarc,myGreen,line width=1pt]    (q5) to (q6);
            \draw[myarc,myRed,line width=1pt] (q6) to (q7);

            \draw[myarc,myBlue, out=155, in=25]    (q3) to (q1);
            \draw[myarc,myBlue, out=150, in=30]    (q4) to (q1);
            \draw[myarc,myBlue, out=145, in=35]    (q5) to (q1);
            \draw[myarc,myBlue, out=140, in=40]    (q6) to (q1);
            \draw[myarc,myBlue, out=135, in=45]    (q7) to (q1);

            \draw[myarc,myBlue, out=155, in=25]    (q4) to (q2);
            \draw[myarc,myBlue, out=150, in=30]    (q5) to (q2);
            \draw[myarc,myBlue, out=145, in=35]    (q6) to (q2);
            \draw[myarc,myBlue, out=140, in=40]    (q7) to (q2);

            \draw[myarc,myBlue, out=155, in=25]    (q5) to (q3);
            \draw[myarc,myBlue, out=150, in=30]    (q6) to (q3);
            \draw[myarc,myBlue, out=145, in=35]    (q7) to (q3);

            \draw[myarc,myBlue, out=155, in=25]    (q6) to (q4);
            \draw[myarc,myBlue, out=150, in=30]    (q7) to (q4);

            \draw[myarc,myBlue, out=155, in=25]    (q7) to (q5);
        \end{tikzpicture}
        }
        \caption{The tournament based on a path of
          length~$6$ that is the majority graph of the election given by \eqref{eq:PP:Schulze:threeVoters:NPh:path_profile} for~$t=7$.
        An arc~$(q_i,q_j)$ in \textcolor{myGreen}{\bfseries green} means that votes~$v_1$ and~$v_2$ (but not~$v_3$) prefer~$q_i$ to~$q_j$;
        an arc~$(q_i,q_j)$ in \textcolor{myRed}{\bfseries red} means that~$v_2$ and~$v_3$ (but not~$v_1$) prefer~$q_i$ to~$q_j$; and finally,
        an arc~$(q_i,q_j)$ in \textcolor{myBlue}{\bfseries blue} means that~$v_1$ and~$v_3$ (but not~$v_2$) prefer~$q_i$ to~$q_j$.}
        \label{fig:PP:Schulze:threeVoters:NPh:tournament_path}
    \end{figure}
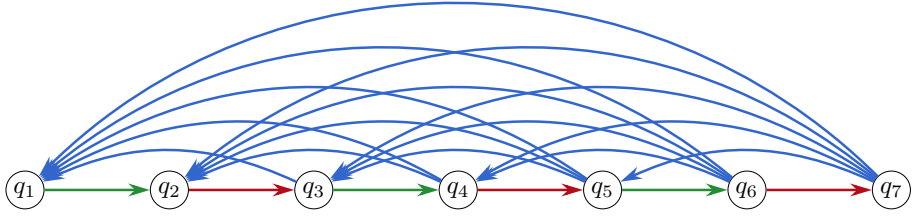

    \proofsubparagraph{Adding an Alternative Path of Length Three.}
    Next, we show how to realize certain series-parallel graphs as majority graphs based on~$T_Q$.
    More precisely, we will show how to modify the preference profile~$\Prfl_Q$ in order to attach a path~$(\pathc_i,\pathc'_{i+1},\pathc'_{i+2},\pathc_{i+3})$ to~$T_Q$ using newly introduced candidates~$\pathc'_{i+1}$ and~$\pathc'_{i+2}$.
    We can think of such an operation as introducing an ``alternative path'' from~$\pathc_i$ to~$\pathc_{i+3}$.
    Notice that~$\Prfl_Q$ contains exactly one vote in~$\{v_1,v_3\}$ whose preference list contains~$\pathc_{i+1}$ and $\pathc_{i+2}$ consecutively (depending on the parity of~$i$), while the other vote in~$\{v_1,v_3\}$ does \emph{not} contain these two candidates consecutively and, in fact, prefers~$\pathc_{i+2}$ to~$\pathc_{i+1}$.
    We perform the following modification on~$\Prfl_Q$:
    \begin{itemize}
    \item For~$v \in \{v_1,v_3\}$,
      if the preferences of~$v$ contain~$\pathc_{i+1}\ \pathc_{i+2}$ consecutively, then we replace them with the series~$\pathc'_{i+1}\ \pathc'_{i+2}\ \pathc_{i+1}\ \pathc_{i+2}$;
        otherwise, we simply insert~$\pathc'_{i+1}$ right after~$\pathc_{i+1}$ and insert~$\pathc'_{i+2}$ right after~$\pathc_{i+2}$.
      \item In the preferences of~$v_2$, we replace~$\pathc_{i+1}\ \pathc_{i+2}$ with
        $\pathc_{i+1}\ \pathc_{i+2}\ \pathc'_{i+1}\ \pathc'_{i+2}$.
    \end{itemize}
    Notice that~$\pathc_i$ defeats both~$\pathc_{i+1}$ and~$\pathc'_{i+1}$, and~$\pathc_{i+3}$ gets defeated by both~$\pathc_{i+2}$ and~$\pathc'_{i+2}$.
    Moreover,~$\pathc'_{i+1}$ defeats~$\pathc'_{i+2}$ because~$\pathc_{i+1}$ defeats~$\pathc_{i+2}$; however, $\pathc_{i+1}$ does \emph{not} defeat~$\pathc'_{i+2}$, and neither does~$\pathc'_{i+1}$ defeat~$\pathc_{i+2}$.
    Hence, all paths from~$\pathc_i$ to~$\pathc_{i+3}$ must either use both~$\pathc_{i+1}$ and~$\pathc_{i+2}$, or they must use both~$\pathc'_{i+1}$ and~$\pathc'_{i+2}$---this will be the key property of this type of modification step.
    See \Cref{fig:PP:Schulze:threeVoters:NPh:adding_path-l3} for an illustration.

    \begin{figure}[bt!]
        \centering
        \begin{tikzpicture}[every node/.style={draw,circle,inner sep=0pt, minimum width=15pt}]
            \node (q1) at (0,0) {$q_1$};
            \node (q2) at (2,1) {$q_2$};
            \node (q3) at (4,1) {$q_3$};
            \node (q4) at (6,0) {$q_4$};
            \node (q'2) at (2,-1) {$q'_2$};
            \node (q'3) at (4,-1) {$q'_3$};

            \draw[myarc,myGreen,line width=1pt]    (q1) to (q2);
            \draw[myarc,myRed,line width=1pt] (q2) to (q3);
            \draw[myarc,myGreen,line width=1pt]    (q3) to (q4);

            \draw[myarc,myGreen,line width=1pt]    (q1) to (q'2);
            \draw[myarc,myRed,line width=1pt] (q'2) to (q'3);
            \draw[myarc,myGreen,line width=1pt]    (q'3) to (q4);

            \draw[myarc,myGreen,line width=1pt]    (q2) to (q'2);
            \draw[myarc,myGreen,line width=1pt]    (q3) to (q'3);
            \draw[myarc,myBlue,line width=1pt]     (q3) to (q'2);
            \draw[myarc,myBlue,line width=1pt]     (q'3) to (q2);

            \draw[myarc,myBlue,line width=1pt]     (q3) to (q'2);
            \draw[myarc,myBlue,line width=1pt]     (q'3) to (q2);

            \draw[myarc,myBlue,line width=1pt,bend right=50] (q3) to (q1);
            \draw[myarc,myBlue,line width=1pt,bend left=50]  (q'3) to (q1);
            \draw[myarc,myBlue,line width=1pt,bend right=50] (q4) to (q2);
            \draw[myarc,myBlue,line width=1pt,bend left=50]  (q4) to (q'2);
            \draw[myBlue,line width=1pt,in=0,out=110]  (q4) to (3,2);
            \draw[myarc,myBlue,line width=1pt,in=70,out=180]  (3,2) to (q1);

            \node[draw=none] (prf) at (9,0) {
                \begin{tabular}{ll}
                  ~$v_1$ :&~$\pathc_3\ \pathc'_3\ \pathc_4\ \pathc_1\ \pathc_2\ \pathc'_2$ \\[4pt]
                  ~$v_2$ :&~$\pathc_1\ \pathc_2\ \pathc_3\ \pathc'_2\ \pathc'_3\ \pathc_4$ \\[4pt]
                  ~$v_3$ :&~$\pathc_4\  \pathc'_2\  \pathc'_3\  \pathc_2\  \pathc_3\  \pathc_1$
                \end{tabular}
            };
        \end{tikzpicture}
        \caption{The majority graph of an election obtained by adding an alternative path~$(\pathc_1,\pathc'_2,\pathc'_3,\pathc_4)$ to a tournament based on the path~$(\pathc_1,\pathc_2,\pathc_3,\pathc_4)$. The meaning of colored arcs is the same as in \Cref{fig:PP:Schulze:threeVoters:NPh:tournament_path}.}
        \label{fig:PP:Schulze:threeVoters:NPh:adding_path-l3}
    \end{figure}
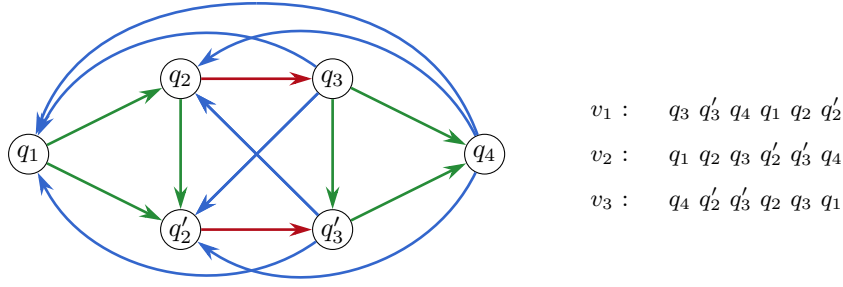

    \proofsubparagraph{Adding Two Alternative Paths of Length Two.}
    We will also add another type of modification where we add not one but two alternative paths at once; however, we will only do this for subpaths of~$Q$ of length two (instead of three as above). More precisely, we now show how to modify the preference profile~$\Prfl_Q$ so that paths~$(\pathc_i,\pathc'_{i+1},\pathc_{i+2})$ and~$(\pathc_i,\pathc''_{i+1},\pathc_{i+2})$ are attached to~$T_Q$ using newly introduced candidates~$\pathc'_{i+1}$ and~$\pathc''_{i+1}$.
    To achieve this, we simply insert~$\pathc'_{i+1}$ and~$\pathc''_{i+1}$ right after~$\pathc_{i+1}$ in any order; then~$\pathc_i$ defeats each of~$\pathc_{i+1},\pathc'_{i+1}$, and~$\pathc''_{i+1}$, and in return, each of these candidates defeat~$\pathc_{i+2}$.
    In particular, all paths from~$\pathc_i$ to~$\pathc_{i+2}$ contain (at least) one of the candidates~$\pathc_{i+1},\pathc'_{i+1}$, and~$\pathc''_{i+1}$.
    See \Cref{fig:PP:Schulze:threeVoters:NPh:adding_paths-l2} for an illustration.

    \begin{figure}[bt!]
        \centering
        \begin{tikzpicture}[
            every node/.style={draw,circle,inner sep=0pt, minimum width=15pt}]
            \node[draw=none] (qdummy) at (0,0) {\phantom{$q_1$}};
            \node (q1) at (1,0) {$q_1$};
            \node (q2) at (3,1) {$q_2$};
            \node (q3) at (5,0) {$q_3$};
            \node (q'2) at (3,0) {$q'_2$};
            \node (q''2) at (3,-1) {$q''_2$};

            \draw[myarc,myGreen,line width=1pt] (q1) to (q2);
            \draw[myarc,myRed,line width=1pt] (q2) to (q3);

            \draw[myarc,myGreen,line width=1pt]    (q1) to (q'2);
            \draw[myarc,myGreen,line width=1pt]    (q1) to (q''2);
            \draw[myarc,myRed,line width=1pt] (q'2) to (q3);
            \draw[myarc,myRed,line width=1pt] (q''2) to (q3);

            \draw[myarc,line width=2pt]    (q2) to (q'2);
            \draw[myarc,line width=2pt]    (q'2) to (q''2);
            \draw[myarc,line width=2pt,bend left=40]    (q2) to (q''2);

            \draw[myBlue,line width=1pt,in=0, out=120] (q3) to (3,1.6);
            \draw[myarc,myBlue,line width=1pt,out=180, in=60] (3,1.6) to (q1);

            \node[draw=none] (prf) at (9,0) {
                \begin{tabular}{ll}
                   $v_1$ :&~$\pathc_3\ \pathc_1\ \pathc_2\ \pathc'_2\ \pathc''_2$\phantom{\ $\pathc_x$} \\[4pt]
                   $v_2$ :&~$\pathc_1\ \pathc_2\ \pathc'_2\ \pathc''_2\ \pathc_3$ \\[4pt]
                   $v_3$ :&~$\pathc_2\  \pathc'_2\  \pathc''_2\  \pathc_3\  \pathc_1$
                \end{tabular}
            };
        \end{tikzpicture}
        \caption{The majority graph of an election obtained by adding alternative paths~$(\pathc_1,\pathc'_2,\pathc_3)$ and~$(\pathc_1,\pathc''_2,\pathc_3)$ to a tournament based on the path~$(\pathc_1,\pathc_2,\pathc_3)$.
        The meaning of colored arcs is the same as in \Cref{fig:PP:Schulze:threeVoters:NPh:tournament_path}, while \textbf{bold, black} arcs represent arcs with weight~$3$.}
        \label{fig:PP:Schulze:threeVoters:NPh:adding_paths-l2}
    \end{figure}
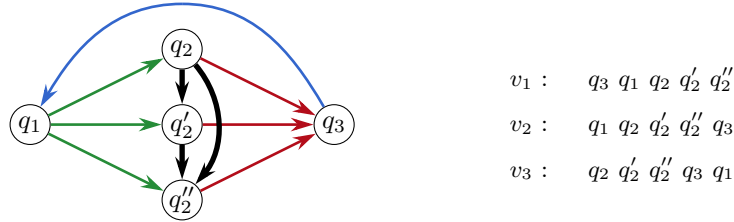

    \proofsubparagraph{Constructing the Reduction.}
    Let~$\varphi$ be our input formula for \BalancedSat, defined over variable set~$X=\{x_1,\dots,x_n\}$, and containing clauses
    $\clause_1,\dots,\clause_m$.
    We are going to construct an election~$\election$.
    We first create~$n$ singleton \emph{variable dummies}~$a_1,\dots,a_n$
    and~$m+1$ singleton \emph{clause dummies}~$b_1,\dots,b_{m+1}$; for ease of notation,  we will write~$a_{n+1}=b_1$.
    Without loss of generality, we may assume that~$n$ is even (as otherwise we can just create a copy~$\varphi'$ of~$\varphi$ on new variables and take~$\varphi \wedge \varphi'$).

    Next, for each variable~$x_i \in X$, we introduce a \emph{variable gadget}~$G_i$ that involves candidates~$a_i$ and~$a_{i+1}$ together with four \emph{variable selection candidates}~$y_i^1,y_i^2,\ol{y}_i^1,\ol{y}_i^2$. We will ensure that the majority graph induced by these candidates will be as shown in the left part of \Cref{fig:PP:Schulze:threeVoters:NPh:gadgets} (which except for the vertex names is the same as \Cref{fig:PP:Schulze:threeVoters:NPh:adding_path-l3}); that is, within gadget~$G_i$ we have that
    \begin{itemize}
        \item~$a_i$ defeats only~$y_i^1$ and~$\ol{y}_i^1$,
        \item~$y_i^1$ defeats only~$y_i^2$ and~$\ol{y}_i^1$,
        \item~$\ol{y}_i^1$ defeats only~$\ol{y}_i^2$,
        \item~$y_i^2$ defeats exactly~$\ol{y}_i^1$,~$\ol{y}_i^2$,~$a_{i+1}$, and~$a_i$
        \item~$\ol{y}_i^2$ defeats exactly~$a_{i+1}$,~$a_i,$ and~$y_i^1$, and finally,
        \item~$a_{i+1}$ defeats exactly~$a_i$,~$y_i^1$, and~$\ol{y}_i^1$.
    \end{itemize}
    In particular, each path within~$G_i$ from~$a_i$ to~$a_{i+1}$  uses either both~$y_i^1$ and~$y_i^2$, or both~$\ol{y}_i^1$ and~$\ol{y}_i^2$.

    \begin{figure}[bt!]
        \centering
        \scalebox{0.87}{
        \begin{tikzpicture}[
            every node/.style={draw,circle,inner sep=0pt, minimum width=20pt}]
            \node[draw=none]  at (-0.2,2) {$G_i\colon$};

            \node (q1) at (0,0) {$a_i$};
            \node (q2) at (2,1) {$y_i^1$};
            \node (q3) at (4,1) {$y_i^2$};
            \node (q4) at (6,0) {$a_{i+1}$};
            \node (q'2) at (2,-1) {$\ol{y}_i^1$};
            \node (q'3) at (4,-1) {$\ol{y}_i^2$};

            \draw[myarc,myGreen,line width=1pt]    (q1) to (q2);
            \draw[myarc,myRed,line width=1pt] (q2) to (q3);
            \draw[myarc,myGreen,line width=1pt]    (q3) to (q4);

            \draw[myarc,myGreen,line width=1pt]    (q1) to (q'2);
            \draw[myarc,myRed,line width=1pt] (q'2) to (q'3);
            \draw[myarc,myGreen,line width=1pt]    (q'3) to (q4);

            \draw[myarc,myGreen,line width=1pt]    (q2) to (q'2);
            \draw[myarc,myGreen,line width=1pt]    (q3) to (q'3);
            \draw[myarc,myBlue,line width=1pt]     (q3) to (q'2);
            \draw[myarc,myBlue,line width=1pt]     (q'3) to (q2);

            \draw[myarc,myBlue,line width=1pt]     (q3) to (q'2);
            \draw[myarc,myBlue,line width=1pt]     (q'3) to (q2);

            \draw[myarc,myBlue,line width=1pt,bend right=50] (q3) to (q1);
            \draw[myarc,myBlue,line width=1pt,bend left=50]  (q'3) to (q1);
            \draw[myarc,myBlue,line width=1pt,bend right=50] (q4) to (q2);
            \draw[myarc,myBlue,line width=1pt,bend left=50]  (q4) to (q'2);
            \draw[myBlue,line width=1pt,in=0,out=110]  (q4) to (3,2);
            \draw[myarc,myBlue,line width=1pt,in=70,out=180]  (3,2) to (q1);

            \node[draw=none]  at (7.8,2) {$H_j\colon$};

            \node (q1) at (8,0) {$b_j$};
            \node (q2) at (10.5,1.3) {$\ell_j^1$};
            \node (q3) at (13,0) {$b_{j+1}$};
            \node (q'2) at (10.5,0) {$\ell_j^2$};
            \node (q''2) at (10.5,-1.3) {$\ell_j^3$};

            \draw[myarc,myGreen,line width=1pt]    (q1) to (q2);
            \draw[myarc,myRed,line width=1pt] (q2) to (q3);

            \draw[myarc,myGreen,line width=1pt]    (q1) to (q'2);
            \draw[myarc,myGreen,line width=1pt]    (q1) to (q''2);
            \draw[myarc,myRed,line width=1pt] (q'2) to (q3);
            \draw[myarc,myRed,line width=1pt] (q''2) to (q3);

            \draw[myarc,line width=2pt]    (q2) to (q'2);
            \draw[myarc,line width=2pt]    (q'2) to (q''2);
            \draw[myarc,line width=2pt,bend left=40]    (q2) to (q''2);

            \draw[myBlue,line width=1pt,in=0, out=120] (q3) to (10.5,2);
            \draw[myarc,myBlue,line width=1pt,out=180, in=60] (10.5,2) to (q1);
        \end{tikzpicture}
        }
        \caption{Variable and clause gadgets in the proof of \Cref{thm:PP:Schulze:threeVoters:NPh}.
          Notation is the same as in \Cref{fig:PP:Schulze:threeVoters:NPh:tournament_path,fig:PP:Schulze:threeVoters:NPh:adding_paths-l2}.
        }
        \label{fig:PP:Schulze:threeVoters:NPh:gadgets}
    \end{figure}
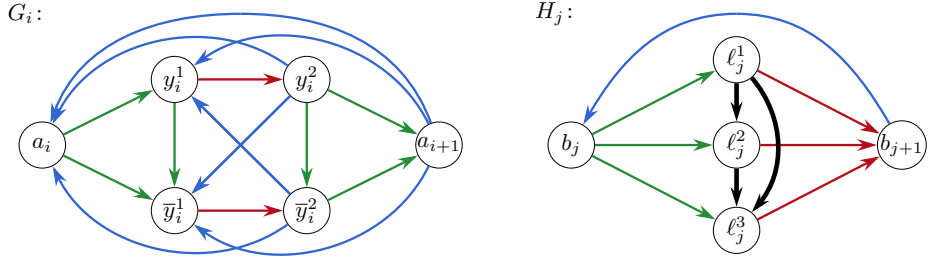

    Next, for each clause~$\clause_j$, we introduce a \emph{clause gadget}~$H_j$ that involves candidates~$b_j$ and~$b_{j+1}$, as well as three \emph{literal occurrence candidates}, $\ell_j^1$,~$\ell_j^2$, and~$\ell_j^3$, corresponding to the three literals occurring in~$\clause_j$; for simplicity, we will identify these three candidates with the corresponding literal occurrences.
    Hence, the set of all literal occurrence candidates will be~$\{x^1_i,x_i^2,\ol{x}_i^1,\ol{x}_i^2:i \in [n]\}$.
    We will ensure that the majority graph induced by the five candidates in the clause gadget~$H_j$ will be as shown in the right part of \Cref{fig:PP:Schulze:threeVoters:NPh:gadgets} (which except for the vertex names is the same as \Cref{fig:PP:Schulze:threeVoters:NPh:adding_paths-l2}); that is, within gadget~$H_i$ we have that
    \begin{itemize}
        \item~$b_j$ defeats each of~$\ell_j^1$,~$\ell_j^2$, and~$\ell_j^3$,
        \item~$\ell_j^1$ defeats~$\ell_j^2$,~$\ell_j^3$, and~$b_{j+1}$,
        \item~$\ell_j^2$ defeats~$\ell_j^3$ and~$b_{j+1}$,
        \item~$\ell_j^3$ defeats only~$b_{j+1}$, and finally,
        \item~$b_{j+1}$ defeats~$b_j$.
    \end{itemize}
    In particular, each path from~$b_j$ to~$b_{j+1}$ uses
    at least
    one of~$\ell_j^1$,~$\ell_j^2$, and~$\ell_j^1$.

    The gadgets~$G_1,\dots,G_n,H_1,\dots,H_{m}$ will be chained one after the other, connected via the variable and clause dummies.
    We will ensure the following property:
    \begin{itemize}
        \item[$(\mypropsymbol)$] \raisebox{\ht\strutbox}{\hypertarget{prop:chain}{}}
        For any two (distinct)  candidates~$c$ and~$c'$ contained in different gadgets, if the gadget containing~$c$ precedes the one containing~$c'$ in the above ordering, then candidate~$c'$ will defeat candidate~$c$.
    \end{itemize}
    This property completes the definition of our majority graph~$G$, since we have completely characterized all relations between candidates within gadgets and between candidates in different gadgets.

    Instead of directly defining the preference profile that yields~$G$ as its a majority graph, we will construct this profile using the machinery developed at the beginning of the proof. To begin with, we construct a profile whose majority graph is the tournament based on the path
    \begin{eqnarray*}
        Q_0(a_1,y_1^1,y_1^2,a_2,y_2^1,y_2^2,a_3,\dots,a_n,y_n^1,y_n^2,a_{n+1}= b_1,\ell_1^1,b_2,\ell_2^1,\dots,b_m,\ell_m^1,b_{m+1}).
    \end{eqnarray*}
    Then, repeatedly applying our method for adding an alternative path of length three, we attach the paths~$(a_i,\ol{y}_i^1,\ol{y}_i^2,a_{i+1})$ for each~$i \in [n]$.
    Next, repeatedly applying our method for adding two alternative paths of length two, we attach the paths~$(b_j,\ell_j^2,b_{j+1})$ and~$(b_j,\ell_j^3,b_{j+1})$ for each~$j \in [m]$.
    Notice that since the sets~$\{y_i^1,y_i^2\}$ for~$i \in [n]$ and the sets~$\{\ell_j^1\}$ for~$j \in [m]$ are pairwise disjoint, these operations can be applied in any order, and will result in a graph that will ensure the relations within each variable and clause gadget as required and will also ensure property~\propchain.

    This completes the definition of our election~$\election$. For convenience, we also give the preference profile for~$\election$:
    \begin{align*}
        v_1\colon&\quad b_{m+1}\ b_m\  \ell_m^1\ \ell_m^2\ \ell_m^3\
        b_{m-1}\  \ell_{m-1}^1\ \ell_{m-1}^2\ \ell_{m-1}^3\
        \cdots\
        b_1=a_{n+1}\  \ell_1^1\ \ell_1^2\ \ell_1^3\
        \\
        & \phantom{m} \ol{y}_n^1\ \ol{y}_n^2\ y_n^1\ y_n^2\ y_{n-1}^2\ \ol{y}_{n-1}^2\  a_n\
        \cdots\  a_3\  y_3^1\ \ol{y}_3^1\ \ol{y}_2^1\ \ol{y}_2^2\ y_2^1\ y_2^2\ y_1^2\ \ol{y}_1^2\  a_2\ a_1\ y_1^1\ \ol{y}_1^1;
        \\
        v_2\colon&\quad a_1\  y_1^1\ y_1^2\ \ol{y}_1^1\ \ol{y}_1^2\  a_2\  y_2^1\ y_2^2\ \ol{y}_2^1\ \ol{y}_2^2\  a_3\ \cdots\  a_n\  y_n^1\ y_n^2\ \ol{y}_n^1\ \ol{y}_n^2\  \\
        & \phantom{m} b_1=a_{n+1}\  \ell_1^1\ \ell_1^2\ \ell_1^3\
        b_2\  \ell_2^1\ \ell_2^2\ \ell_2^3\ \cdots\
        b_m\  \ell_m^1\ \ell_m^2\ \ell_m^3\  b_{m+1};
        \\
        v_3\colon&\quad
        \ell_m^1\ \ell_m^2\ \ell_m^3\  b_{m+1}\
        \ell_{m-1}^1\ \ell_{m-1}^2\ \ell_{m-1}^3\  b_m\  \cdots\
         \ell_1^1\ \ell_1^2\ \ell_1^3\  b_2\  \\
        & \phantom{m}
        y_n^2\  \ol{y}_n^2\ a_{n+1}=b_1\ a_n\ y_n^1\ \ol{y}_n^1\
        \cdots\  y_2^2\  \ol{y}_2^2\ a_3\ a_2\ y_2^1\ \ol{y}_2^1\ \ol{y}_1^1\ \ol{y}_1^2\ y_1^1\ y_1^2\ a_1.
    \end{align*}

    Completing the definition of our instance of \SchPP,  we set~$\{a_1\}$ as the distinguished party and define the party structure as follows.
    As already mentioned, all variable and clause dummies will form their own singleton party.
    Hence, it remains to partition the variable selection candidates and the literal occurrence candidates into parties.
    To do so, for each variable~$x_i \in X$ we form four parties:
    $\{x_i^1,y_i^1\}$,
    $\{x_i^2,y_i^2\}$,
    $\{\ol{x}_i^1,\ol{y}_i^1\}$, and
    $\{\ol{x}_i^2,\ol{y}_i^2\}$.
    This completes the definition of our instance; note that the maximum party size is two.

    \proofsubparagraph{Correctness.}
    First, assume that~$\varphi$ admits a satisfying truth assignment. Let~$T$ denote the set of true literals, and let~$Z_T=\{z^1,z^2\}$ denote the set of the corresponding literal occurrence candidates.
    Construct the set~$\nomination$ of nominees as follows: Besides all singletons,~$\nomination$ contains all candidates in~$Z_T$, as well as the variable selection candidate from each party of size two whose literal occurrence candidate is \emph{not} in~$Z_T$.
    That is, if~$x_i$ is a true literal, then~$\nomination$ contains~$x_i^1$,~$x_i^2$,~$\ol{y}_i^1$, and~$\ol{y}_i^2$, and otherwise it contains the candidates~$\ol{x}_i^1$,~$\ol{x}_i^2$,~$y_i^1$, and~$y_i^2$.

    Observe that the majority subgraph of the reduced election~$\election(N)$ contains a path from~$a_i$ to~$a_{i+1}$ for each~$i \in [n]$ through the two nominated variable selection candidates associated with~$x_i$, and moreover, it contains a path from~$b_j$ to~$b_{j+1}$ for each~$j \in [m]$ through the literal occurrence candidate associated with some true literal that occurs in the clause~$\clause_j$ (if there are two or three true literals in~$\clause_j$, there is such a path through each of them).
    This means that there is a path in~$\majGN$ from~$a_1$ to every nominee in~$\nomination$, and therefore the beatpath strength~$\beats(a_1,c)$  of~$a_1$ over each nominee~$c \in N$ is exactly~$1$ (note that~$a_1$ has no outgoing arcs of weight~$3$ in~$\majG$).
    Since no candidate is preferred to~$a_1$ by all voters (as~$a_1$ is the top choice of~$v_2)$, it follows that~$\beats(c,a_1)\leq 1$ for each nominee~$c \in N$.
    Hence,~$a_1$ is a winner in~$\election(N)$.

    For the other direction, note that~$a_1$ is defeated by~$b_{m+1}$, so for~$a_1$ to be a winner in a reduced election over a set of nominees, the majority graph~$\majGN$ must contain a path from~$a_1$ to~$b_{m+1}$.
    This path must go through all variable and clause dummies, and in particular, must subsume a path from~$a_i$ to~$a_{i+1}$ for each~$i \in [n]$ and a path from~$b_j$ to~$b_{j+1}$ for each~$j \in [m]$.
    Due to the properties of variable gadgets, this means that~$\nomination$ must either contain both~$\ol{y}_i^1$ and~$\ol{y}_i^2$, or it must contain both~$y_i^1$ and~$y_i^2$.
    In the former case holds, we set~$x_i$ to \true, otherwise we set~$x_i$ to \false.
    We claim that the resulting truth assignment~$\alpha$ satisfies~$\varphi$.
    Indeed, if~$\alpha(x_i)=\true$, then (as both~$\ol{y}_i^1$ and~$\ol{y}_i^2$ are nominated) neither~$\ol{x}_i^1$ nor~$\ol{x}_i^2$ is nominated in~$\election(N)$ due to the party structure; similarly, if~$\alpha(x_i)=\false$, then (as both~$y_i^1$ and~$y_i^2$ are nominated) neither~$x_i^1$ nor~$x_i^2$ is nominated.
    Therefore, all nominated literal occurrence candidates correspond to literals that evaluate to \true{} under~$\alpha$.
    Now, since for each~$j \in [m]$, all paths from~$b_j$ to~$b_{j+1}$ must contain at least one literal occurrence candidate that corresponds to a literal contained in the clause~$\clause_j$, it follows that each clause~$\clause_j$ contains at least one literal that is set to \true\ under~$\alpha$.
    Hence,~$\varphi$ is satisfied by~$\alpha$, proving the correctness of our reduction for the case of three voters.
    %
\end{proof}

\begin{theorem}
\label{thm:PP:Schulze:fourVoters:NPh}
  \SchPP is \NPc, even if no party contains more than two candidates and the number
  of voters is an even constant at least four.
\end{theorem}

\begin{proof}
    We again present a reduction from \BalancedSat.
    Let~$\varphi=\bigwedge_{j \in[m]} \clause_j$ be the input 3-CNF formula defined over a set~$X=\{x_1,\dots,x_n\}$ of variables.
    We let~$\clauses = \{\clause_1,\dots,\clause_m\}$ denote the set of clauses in~$\varphi$.

    \proofsubparagraph{Construction.}
    Let us create an election~$\election$ with four voters.
    We let~$\{p\}$ be the distinguished party.
    For each variable~$x_i \in X$, we create a party~$\{x_i,\overline{x}_i\}$ containing two \emph{variable candidates} corresponding to the positive and negative form of~$x_i$.
    For each clause~$\clause_j \in \clauses$, we create a singleton party containing the \emph{clause candidate}~$\clause_j$.
    We further create four \emph{literal occurrence candidates},~$x_i^1$,~$x_i^2$,~$\ol{x}_i^1$, and~$\ol{x}_i^2$, corresponding to the four occurrences of each variable~$x_i \in X$; these candidates are partitioned into parties~$\{x_i^1,\ol{x}_i^1\}$ and~$\{x_i^2,\ol{x}_i^2\}$.
    This completes the definition of the candidate set and the party structure in~$\election$.
    Note that, as promised, each party has at most two candidates.

    We will use the following notation.
    Let~$\ol{X}=\{\ol{x}_1,\dots,\ol{x}_n\}$ denote the set of all negated variables,
    so that~$X \cup \ol{X}$ is the set of all variable candidates.\footnote{With a slight abuse of notation, we identify variables and their negated forms, as well as clauses, with the corresponding candidates.}
    Let~$L=\{x_i^h,\ol{x}_i^h:i \in [n],\ h \in [2]\}$ denote the set of all literal occurrence candidates.

    We are now ready to present the preferences of our four votes,~$v_1,\dots,v_4$:
    \begin{align*}
        v_1\colon&\quad \ora{\clauses}\   p\
        x_1\ x_1^1\ x_1^2\  \cdots\
        x_n\ x_n^1\ x_n^2\
        \ol{x}_1\ \ol{x}_1^1\ \ol{x}_1^2\ \cdots\
        \ol{x}_n\ \ol{x}_n^1\ \ol{x}_n^2; \\[4pt]
        v_2\colon&\quad p\
        \ol{x}_n\ \ol{x}_n^2\ \ol{x}_n^1\
        \cdots\
        \ol{x}_1\ \ol{x}_1^2\ \ol{x}_1^1\
        x_n\ x_n^2\ x_n^1\  \cdots\
        x_1\ x_1^2\ x_1^1\
        \ola{\clauses};\\[4pt]
        v_3\colon&\quad \ora{X \cup \ol{X}}\
        \ell_1^1\ \ell_1^2\ \ell_1^3\  \clause_1\
        \cdots\
        \ell_m^1\ \ell_m^2\ \ell_m^3\  \clause_m\  p;
        \\[4pt]
        v_4\colon&\quad \ell_m^3\ \ell_m^2\ \ell_m^1\  \clause_m\
        \cdots\
        \ell_1^3\ \ell_1^2\ \ell_1^1\  \clause_1\
        p\  \ola{X \cup \ol{X}}.
    \end{align*}

    Before we construct the weighted majority graph~$\majG$ of the election~$\election$ containing all candidates, we make several observations.
    First, note that the preference profile ensures that no two candidates of the same type have an arc between them in~$\majG$.
    There are also no arcs between~$p$ and any literal occurrence candidate, or between any variable candidate and any clause candidate.
    Additionally, since a variable candidate~$x_i$ and a literal occurrence candidate~$x_j^h$ for~$j \neq i$ are ordered oppositely in the votes~$v_1$ and~$v_2$ (as well as in~$v_3$ and~$v_4$), there is no arc between them in~$\majG$.
    Similarly, since any clause candidate~$\clause_j$ and any literal occurrence candidate~$\ell \in L \setminus \{\ell_j^1,\ell_j^2,\ell_j^3\}$ (not occurring in the clause~$\clause_j$) are ordered oppositely in the votes~$v_3$ and~$v_4$ (as well as in~$v_1$ and~$v_2$), there is no arc between them either in~$\majG$.
    Therefore, the only arcs present in~$\majG$ are those listed in \Cref{tab:PP:Schulze:fourVoters:NPh:weights}.

    \begin{table}[bt!]
        \caption{Weights of the arcs in~$\majG$ defined in the proof of \Cref{thm:PP:Schulze:fourVoters:NPh}.}
        \label{tab:PP:Schulze:fourVoters:NPh:weights}
        \centering
        \renewcommand{\arraystretch}{1.3}
        \begin{tabular}{l @{\quad} c @{\quad} c}
          \toprule
             arc~$(a,b)$ & $\weight(a,b)$ & Who prefers~$a$ to~$b$? \\
             \midrule
             $(p,z)$,~$z \in X \cup \ol{X}$ & $2$ &$v_1,v_2,v_4$ \\
             $(z,\ell)$,~$z \in X \cup \ol{X}$,~$\ell=z^h$ for some~$h \in [2]$
             & $2$ &
             $v_1, v_2, v_3$
             \\
             $(\ell_j^h,\clause_j)$,~$\clause_j \in \clauses$,~$h \in [3]$
             & $2$ & $v_2,v_3,v_4$ \\
             $(\clause_j,p)$,~$\clause_j \in \clauses$ & $2$ & $v_1,v_3,v_4$\\
             \bottomrule
        \end{tabular}
    \end{table}

    \proofsubparagraph{Correctness.}
    First assume that~$p$ is a winner in some reduced election~$\election(N)$ over a set~$\nomination$ of nominees.
    Notice that~$\clauses \cup \{p\} \subseteq N$, and therefore,~$\beats_N(\clause_j,p)=2$ due to the arc~$(\clause_j,p)$.
    Therefore,~$p$ can only be a winner in~$\election(N)$ if there is a path from~$p$ to each clause candidate~$\clause_j$ in~$\majG$, consisting of arcs of weight~$2$.
    Due to the structure of~$\majGN$, such a path must take the form~$(p,z,\ell,\clause_j)$ for some~$z \in X \cup \ol{X}$ and some~$\ell \in L$ where, additionally,
    we have that~$\ell$ is an occurrence of the literal~$z$ and~$z$ is a literal
    occurring in clause~$\clause_j$.
    Therefore, setting all nominated literals, i.e., the nominated candidate from the party~$\{x_i,\ol{x}_i\}$ for each~$x_i \in X$
    to \true\ yields a truth assignment that satisfies each clause of~$\varphi$.

    Conversely, assume now that there exists a satisfying truth assignment for~$\varphi$.
    Let~$T \subseteq X \cup \ol{X}$ denote the set of true literals.
    Let us then nominate all variable candidates from~$T$, as well as their two occurrences from~$L$, that is,~$z^1$ and~$z^2$ for each true literal~$z \in T$, along with all candidates in~$\clauses \cup \{p\}$.
    Then~$p$ is a winner of the resulting election:
    There is an arc in~$\majG$ from~$p$ to each variable candidate in~$T$, so by~$\weight(z,z^1)=\weight(z,z^2)=2$, there is a path of strength~$2$ from~$z$ to each nominated literal occurrence candidate in~$\{z^1,z^2:z \in T\}$, and from there to each clause candidate~$\clause_j$ via some true literal occurring in~$\clause_j$.
    Thus~$\beats(p,a)=2$ for each nominee~$a$ (other than itself) in the resulting election, showing the correctness of our reduction for four voters.
\end{proof}


\subsubsection{Parameterization by the Number of Parties}

Now we turn to studying the parameterized complexity of \SchPP, starting with an easy observation.

\begin{observation}\label{thm:PP:Schulze:numParties:XP}\label{thm:PP:Schulze:numParties:partySize:FPT}
    \SchPP is
    \begin{itemize}
        \item in \XP when parameterized by the number~$\numParties$ of parties, and
        \item fixed-parameter tractable when parameterized by the number~$\numParties$ of parties and the maximum party size~$\maxPartySize$, combined.
    \end{itemize}
\end{observation}
\begin{proof}
    We give a simple brute-force algorithm.
    For each of $\numParties$ parties, there are at most $\maxPartySize$ possible nominations.
    Therefore, there are at most~$\maxPartySize^\numParties$ possible nominations $\nomination$, and for each of them we can verify in polynomial time whether the
    candidate from the distinguished party is a winner.
    This algorithm runs in $\maxPartySize^\Oh{\numParties}$ time.
\end{proof}

The above brute-force approach becomes quickly infeasible as $k$ grows. For three voters, we offer a more efficient algorithm. Although the case~$|V|=3$ is still \NPh, as established by \Cref{thm:PP:Schulze:threeVoters:NPh}, we propose an algorithm that solves it efficiently if the number of parties is relatively small.

Our approach for three voters relies on the following observations. First, if the number~$|V|$ of voters in some election~$\electionI$ is odd, then the weighted majority graph $\majG$ of $\election$ is a tournament; i.e., a complete directed graph that contains exactly one of the two arcs $(x,y)$ and $(y,x)$ for each pair $x,y\in C$.
More specifically, for $|V|=3$ voters, the weight of each arc is either~3 or~1.
In addition, we have a similar assertion as in \Cref{obs:Schulze:twoVoters:transitiveMajorityGraph}:

\begin{observation}
    \label{obs:Schulze:threeVoters:transitiveMajorityGraph}
    Let $\electionI$ be an election with $|V|=3$. Then $G_3(\election)$ is transitive.
\end{observation}

\begin{theorem}
    \label{thm:PP:Schulze:threeVoters:numParties:FPT}
    Given an election~$\electionI$ with~$\numVoters=3$,
    \SchPP is fixed-parameter tractable when parameterized by the number~$\numParties$ of parties.
\end{theorem}
\begin{proof}
    Let $(\election,\P,P)$ be our input instance with election~$\electionI$.
    Without loss of generality, we may assume that $P=\{\distc\}$, as otherwise we can guess a candidate in~$P$ that is a possible president and delete all other candidates from~$P$.

    We start by constructing the weighted majority graph~$G=\majG$.
    If $G$ contains an arc $(c,\distc)$ with weight~3, then due to transitivity of $G_3$ as established in \Cref{obs:Schulze:threeVoters:transitiveMajorityGraph}, we obtain $\beats(\distc,c)<3$. Hence, candidate~$\distc$ cannot be a winner in any reduced election whose candidate set contains~$c$. Therefore, we can remove any such candidate~$c$ safely.

    To proceed, we will rely on the following claim:
    \begin{claim}
    \label{clm:PP:Schulze:threeVoters:paths-to-all}
    Suppose that no arc of weight~$3$ enters~$\distc$ in the weighted majority graph~$G$. Then $\distc$ is a winner in the election~$\election(\nomination)$ resulting from a valid set~$\nomination$ of nominations
    if and only if
    there exists a path from~$\distc$ to every other nominee in~$\majGN$.
    \end{claim}
    \begin{claimproof}
      Suppose first that $\distc$ is a winner in a valid set~$N$ of nominees. Then, since $\distc$ is a winner in~$\election(N)$, every other nominee $c \in N \setminus \{\distc\}$ is either defeated by~$\distc$, in which case the arc~$(\distc,c)$ is the desired path, or defeats~$\distc$, in which case
      $\weight(c,\distc)=1$ yields $\beats(\distc,c) = 1$ (as $\distc$ is a winner), implying the existence of the desired path in~$\majGN$.

        For the other direction it suffices to observe that if there is a path from~$\distc$ to every other nominee~$c$ in~$N$, then $\beats(\distc,c)\geq 1$. However, as there is no arc of weight~$3$ entering~$\distc$ in~$G$, we get $\beats(c,\distc)\leq 1$. Hence, $\distc$ is indeed a winner in~$\election(N)$.
    \end{claimproof}

    \def\comp{\mathsf{comp}}

    Let us assume that $\nomination$ is a valid set of nominations in which $\distc$ is a winner. By \Cref{clm:PP:Schulze:threeVoters:paths-to-all}, there exists a path in~$G[N]$ from~$\distc$ to each nominee in~$\nomination$. Then there exists a directed tree~$T$ over~$N$ rooted at~$\distc$ in~$G[N]$ in which all nominees in~$N\setminus \{\distc\}$ have in-degree exactly one. Let us fix such a directed tree~$T^N$; note that the vertices of~$T^N$ are nominees.
    We will say that a rooted directed tree~$T$ defined over~$\P$ is the \emph{party ideal} of~$T^N$ if the function that maps each candidate in~$N$ to its party is an isomorphism from~$T^N$ to~$T$.

    As our next step, we guess the party ideal~$T$ of~$T^N$; note that there are at most~$\numParties^\numParties$ for~$T$, because $T$ must be a directed out-tree rooted at~$P$.
    We will use the following notation:
    For any party~$Q$, we let $T_Q$ denote the subtree of~$T$ rooted at~$Q$, and we let $\P_Q$ denote the parties occurring in~$T_Q$.

    Next, we will use bottom-up dynamic programming on~$T$ in order to find a valid set~$\nomination'$ of nominations that allow $\distc$ to become a winner.
    In particular, for each party~$Q \in \P$, we compute the set~$\comp_T(Q)$ of candidates in~$Q$ that are \emph{compatible with $T$}, meaning the following: There exists a set~$N_Q$ of nominations by the parties in~$T_Q$ such that the subgraph $G[N_Q]$ of the majority graph induced by these nominees contains a directed tree whose party ideal is~$T$.
    To compute $\comp_T(Q)$, we distinguish between the following two cases:
\begin{description}
    \item[Leaf node.] If $Q$ is a leaf in~$T$, then by definition $\comp_T(Q)=Q$, since each candidate is compatible with~$T$.

    \item[Internal node.] If $Q$ is a node in~$T$ with children $Q_1,\dots,Q_\ell$, then let $S_1,\dots,S_\ell$ denote the corresponding sets of candidates compatible with~$T$, i.e., $S_i=\comp_T(Q_i)$ for each $i \in [\ell]$.
      Then a candidate~$q \in Q$ is compatible with~$T$ if and only if for each $i \in [\ell]$, there exists a candidate~$s_i \in S_i$ such that $(q,s_i)$ is an arc in~$G$.
      This allows us to compute $\comp_T(Q)$ in time that is linear in the total number of arcs leaving some candidate from~$Q$ in~$G$.
\end{description}

Using bottom-up dynamic programming, we can thus compute in~$\Oh{|\candidates|^2}$ time the set $\comp_T(P)$.
Observe that all candidates in~$\comp_T(P)$ are possible presidents; hence, if $\comp_T(P) \neq \emptyset$, then we answer ``yes''; otherwise, we proceed with our next guess for the tree~$T$.
If all possible choices are exhausted without outputting ``yes,'' we output ``no.''
    The total running time of our algorithm is $\numParties^\numParties \cdot \Oh{|C|^2}$, which is FPT with respect to parameter~$\numParties$.
\end{proof}

Having shown fixed-parameter tractability for few parties and three voters, we turn to proving parameterized hardness with respect to the same parameter, the number of parties, for inputs with a large enough even number of voters.

\begin{theorem}
\label{thm:PP:Schulze:numParties:Wh}
  When parameterized by the number~$\numParties$ of parties, it is \Wh to decide \SchPP, even if the number
  of voters is any even constant at least eight.
\end{theorem}

\begin{proof}
  We reduce from the \probName{Multicolored Clique} problem.
  Here, we are given a~$q$-partite graph~$H=(U, F)$ with its vertex set partitioned as~$U=U_1 \cup \cdots \cup U_q$, and the goal is to decide whether a subset~$K\subseteq U$ of vertices of size~$q$ exists so that the vertices of~$K$ induce a complete graph in~$H$.
  The \probName{Multicolored Clique} problem is known to be \Wh when parameterized by~$q$~\cite{fel-her-ros-via:j:multicolored-clique,pie:j:multicolored-clique}.
  We call each~$U_i$,~$i\in[q]$, a \emph{color class}, and we assume without loss of generality that all color classes are of the same size,~$x$, so we can use the
notation~$U_i=\{u_i^1,\dots,u_i^\colorSize\}$ for each
  $i\in[q]$.
  For each~$i,j\in [q]$ with~$i<j$, we denote by~$F_{i,j}$ the set of all edges between the color classes~$U_i$ and~$U_j$, that is,~$F_{i,j} = \{ \{u,w\} : u \in U_i \text{ and } w \in U_j \}$.
  Again, by the construction of Fellows \emph{et al.}~\cite{fel-her-ros-via:j:multicolored-clique}, we can assume that the size of each~$F_{i,j}$ is the same, and we will denote it by~$\edgeSetSize$.
  We will also use the following notation for each
  $i \in [q]$ and each vertex~$u \in U_i$:
    \begin{align*}
      F_{i \to}(u)& =\{f:
      f = \{u,v\} \in F_{i,j} \text{ for some } v \in U_j \text{ with } j>i\}; \\
      F_{\to i}(u)& =\{f:
      f = \{u,v\} \in F_{i,j} \text{ for some } v \in U_j \text{ with } j<i\}.
    \end{align*}

    \proofsubparagraph{Construction.}
    Given an instance~$\mathcal{I} = (H,q)$
    of \probName{Multicolored Clique}, we construct an equivalent instance~$\mathcal{J}$ of \SchPP with election~$\election$ as follows.
    We have three \emph{core candidates},~$a$,~$b$, and~$p$; each of them forms a singleton party,
    and~$\{p\}$ is the distinguished party.
    Next, for each vertex in~$U$, we introduce a \emph{vertex candidate}, also denoted as~$u$, and for each
    $i \in [q]$, we add the color class~$U_i$ as a \emph{color party}.
    Next, for each edge~$f \in F$, we introduce an \emph{edge candidate}, also denoted as~$f$, and
    for each
    $i,j \in [q]$ with~$i<j$, we add~$F_{i,j}$ as an \emph{edge party}.
    This completes the construction of the candidate set~$\candidates$ and
    its partition into parties.

    $\voters$ consists of exactly eight votes: Four \emph{base votes} and four \emph{incidence votes}. We start with a description of the former group.
    These votes create the base of the weighted majority graph and are defined as follows:
    \begin{align*}
        v_1 \colon&\quad
        b\  \ora{U_1\cup\cdots\cup U_q}\  \ora{F_{1,2} \cup \cdots \cup F_{q-1,q}}\  a\  p; \\[4pt]
        v_2 \colon&\quad
        p\  \ola{F_{1,2} \cup \cdots \cup F_{q-1,q}}\  a\  b\  \ola{U_1\cup\cdots\cup U_q}; \\[4pt]
        v_3 \colon&\quad
        \ora{U_1\cup\cdots\cup U_q}\  \ora{F_{1,2} \cup \cdots \cup F_{q-1,q}}\  a\  p\  b;  \\[4pt]
        v_4 \colon&\quad
        b\  a\  \ola{U_1\cup\cdots\cup U_q}\  \ola{F_{1,2} \cup \cdots \cup F_{q-1,q}}\  p.
    \end{align*}

    \begin{figure}[bt!]
        \centering
        \scalebox{0.9}{
        \begin{tikzpicture}[every node/.style={draw,circle,inner sep=0,minimum width=0.8cm}]
            \node (b) at (0,1) {$b$};

            \node[draw=none,font=\footnotesize] at (3,0) {vertex candidates};
            \node (u11) at (-3.5,-1) {$u_1^1$};
            \node[draw=none] (v1dots) at (-2.5,-1) {$\cdots$};
            \node (v1x) at (-1.5,-1) {$u_1^\colorSize$};
            \draw[double,dotted,rounded corners] (-4.1,-0.4) rectangle (-0.9,-1.6);
            \node[draw=none,font=\footnotesize,myGray] at (-0.65,-1.35) {$U_1$};

            \node[draw=none] (vdots) at (0.05,-1) {\huge$\cdots$};

            \node (vq1) at (1.5,-1) {$u_q^1$};
            \node[draw=none] (vqdots) at (2.5,-1) {$\cdots$};
            \node (vqx) at (3.5,-1) {$u_q^\colorSize$};
            \draw[double,dotted,rounded corners] (0.9,-0.4) rectangle (4.1,-1.6);
            \node[draw=none,font=\footnotesize,myGray] at (4.35,-1.4) {$U_q$};

            \draw[dashed,rounded corners] (-4.6,-0.25) rectangle (4.6,-1.75);

            \draw[myarc,myGreen,double] (b) to (0,-0.25);

            \def\ey{-3.5}
            \node[draw=none,font=\footnotesize] at (4.9,-2.5) {edge candidates};
            \node (e121) at (-5,\ey) {$f_{1,2}^1$};
            \node (e121) at (-4,\ey) {$f_{1,2}^2$};
            \node[draw=none] (e12dots) at (-3,\ey) {$\cdots$};
            \node (e12y) at (-2,\ey) {$f_{1,2}^\edgeSetSize$};
            \draw[double,dotted,rounded corners] (-5.5,-3) rectangle (-1.5,-4);
            \node[draw=none,font=\footnotesize,myGray] at (-1.15,-3.85) {$F_{1,2}$};

            \node[draw=none] (edots) at (0.05,\ey) {\huge$\cdots$};

            \node (eqq1) at (2,\ey) {$f_{q',q}^1$};
            \node (eqq1) at (3,\ey) {$f_{q',q}^2$};
            \node[draw=none] (e12dots) at (4,\ey) {$\cdots$};
            \node (eqqy) at (5,\ey) {$f_{q',q}^\edgeSetSize$};
            \draw[double,dotted,rounded corners] (1.5,-3) rectangle (5.5,-4);
            \node[draw=none,font=\footnotesize,myGray] at (5.85,-3.85) {$F_{q',q}$};

            \draw[dashed,rounded corners] (-6.25,-2.75) rectangle (6.25,-4.25);

            \draw[myarc,myRed,double] (0,-1.75) to (0,-2.75);

            \node (a) at (0,-5.5) {$a$};
            \draw[myarc,myBlue,double] (0,-4.25) to (a);

            \node (p) at (0,-7) {$p$};
            \draw[myarc,myRed] (a) to (p);

            \draw[myarc,myRed,double,out=-90,in=120] (-1,-4.25) to (p);
            \draw[line width=7pt,white,out=-90,in=60] (1,-1.75) to (p);
            \draw[myarc,myRed,double,out=-90,in=60] (1,-1.75) to (p);
        \end{tikzpicture}
        }
        \caption{A schematic illustration of the weighted majority graph induced by the base votes~$v_1$,~$v_2$,~$v_3$, and~$v_4$ in the construction of \Cref{thm:PP:Schulze:numParties:Wh}. A double arrow represents that there is a complete bipartite graph between agents of the source and the target group.
          An arc~$(s,t)$ is colored in \textcolor{myGreen}{\bfseries green} if the source candidate~$s$ is preferred in~$v_1$,~$v_2$, and~$v_4$ (but not in~$v_3$); in \textcolor{myRed}{\bfseries red} if~$v_1$,~$v_3$, and~$v_4$ (but not~$v_2$) prefer~$s$ over~$t$; and finally, \textcolor{myBlue}{\bfseries blue} means that~$v_1$,~$v_2$, and~$v_3$ (but not~$v_4$) prefer~$s$ over~$t$.
          We use~$q' = q-1$.
        }
        \label{fig:PP:Schulze:numParties:Wh:initialMajorityGraph}
    \end{figure}
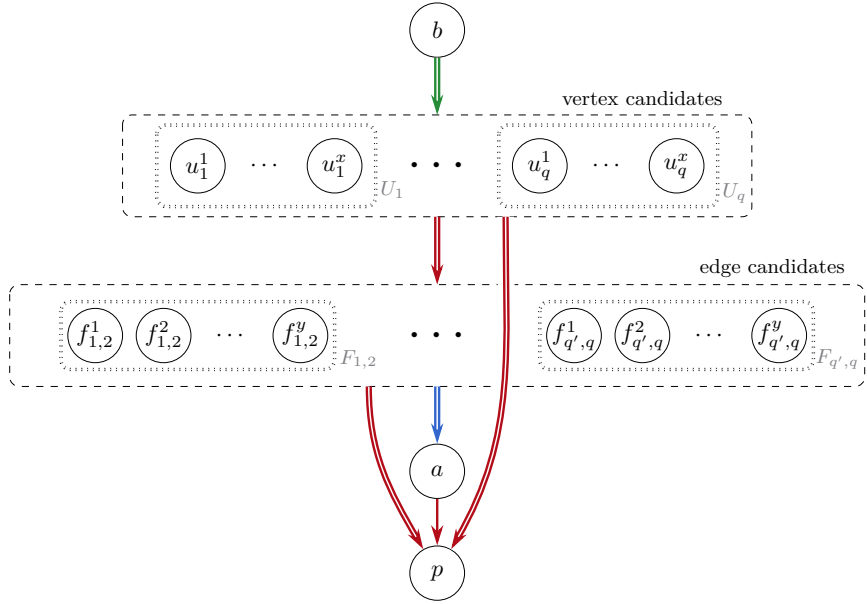

    The weighted majority graph of the election induced by these votes is as follows (see \Cref{fig:PP:Schulze:numParties:Wh:initialMajorityGraph} for an illustration):
    \begin{itemize}
        \item[$\bullet$] candidate~$b$ defeats all vertex candidates;
        \item[$\circ$]  every vertex candidate defeats all edge candidates;
        \item[$\bullet$] every edge candidate defeats candidate~$a$;
        \item[$\bullet$] candidate~$a$ defeats candidate~$p$;
        \item[$\circ$] all vertex and edge candidates defeat candidate~$p$.
    \end{itemize}

    All other combinations of candidates are tied.
    It is also easy to see that the weight of each arc in~$\majG$ is exactly~$2$.
    Observe that, currently, candidate~$b$ has a beatpath of strength~$2$ to all other candidates, which cannot be changed by any nominations.
    Hence,~$b$ is the only winner of this subprofile.
    To remedy this, we add four additional \emph{incidence votes}.
    The preferences of these votes will leave intact the arcs of the majority graph implied by the observations above marked with a bullet~$\bullet$; by contrast, they will change the one marked with~$\circ$.
    To define the preferences of the incidence votes, we will use the following candidate series as building blocks:
    \begin{align*}
        A_{i\to} &:
        \ora{F_{i \to}(u_i^1)}\  u_i^1\
        \ora{F_{i \to}(u_i^2)}\  u_i^2\
        \cdots\
        \ora{F_{i \to}(u_i^\colorSize)}\  u_1^\colorSize;
        \\
        \wt{A}_{i\to} &:
        \ola{F_{i \to}(u_i^\colorSize)}\  u_i^\colorSize\
        \ola{F_{i \to}(u_i^{n-1})}\  u_i^{n-1}\
        \cdots\
        \ola{F_{i \to}(u_i^1)}\  u_1^1;
        \\
        B_{\to i} &: \ora{F_{\to i}(u_i^1)}\  u_i^1\
        \ora{F_{\to i}(u_i^2)}\  u_i^2\
        \cdots\
        \ora{F_{\to i}(u_i^\colorSize)}\  u_1^\colorSize;
        \\
        \wt{B}_{\to i} &:
        \ola{F_{\to i}(u_i^\colorSize)}\  u_i^\colorSize\
        \ola{F_{\to i}(u_i^{n-1})}\  u_i^{n-1}\
        \cdots\
        \ola{F_{\to i}(u_i^1)}\  u_1^1.
    \end{align*}

    Note that~$A_{i \to}$ and~$\wt{A}_{i \to}$ contain the same set of candidates. Moreover,  any two candidates they contain are ordered oppositely by~$A_{i \to}$ and by~$\wt{A}_{i \to}$, \emph{unless they are a vertex candidate~$u \in U_i$ and an incident edge candidate~$f$ that is contained in~$F_{i \to }(u)$}, i.e., $f$ connects~$u \in U_i$  with some vertex in a color class~$U_j$ with a larger index~$j>i$; we refer to this fact as ($\dagger$).

    Similarly,~$B_{\to i}$ and~$\wt{B}_{\to i}$ contain the same set of candidates, they order any two of these candidates in the opposite way,  \emph{unless they are a vertex candidate ${u \in U_i}$ and an incident edge candidate~$f$ that is contained in~$F_{\to i}(u)$}, i.e., $f$ connects~$u \in U_i$ with some vertex in a color class~$U_j$ with a smaller index~$j<i$; we refer to this fact as ($\ddagger$).
    These observations are crucial for our construction.

    We are now ready to give the preferences of our incidence votes:
    \begin{align*}
        \hat{v}_1 \colon& \quad
        p\  A_{1 \to}\  A_{2 \to}\ \cdots\  A_{(q-1) \to}\
        A_{q \to}\  a\  b;
        \\
        \hat{v}_2 \colon& \quad
        b\  a\  p\  \wt{A}_{q \to}\  \wt{A}_{(q-1) \to}\  \cdots\  \wt{A}_{2 \to}\  \wt{A}_{1 \to};
        \\
        \hat{v}_3 \colon& \quad
        p\
        B_{1 \to}\  B_{2 \to}\ \cdots\  B_{(q-1) \to}\
        B_{q \to}\  a\  b;
        \\
        \hat{v}_4 \colon& \quad
        b\  a\  p\  \wt{B}_{q \to}\  \wt{B}_{(q-1) \to}\  \cdots\  \wt{B}_{2 \to}\  \wt{B}_{1 \to}.
    \end{align*}

    Observe that these votes do not change the relations between a core candidate~$c \in \{a,b\}$ and any other candidate~$d$, because $\hat{v}_1$ and~$\hat{v}_2$, as well as~$\hat{v}_3$ and~$\hat{v}_4$, rank~$c$ and~$d$ in the opposite order; hence, the observations marked with~$\bullet$ above remain true.
    Moreover,~$p$ is preferred to all vertex and edge candidates by all incidence votes; hence we have that
    \begin{itemize}
        \item[$\bullet$] $(p,d)$ is an arc of weight~$2$ in~$\majG$ for each vertex or edge candidate~$d$.
    \end{itemize}

    Let us now turn our attention to the relations between a vertex candidate~$u$ and some edge candidate~$f$.
    First,~$\hat{v}_1$ and~$\hat{v}_2$ clearly rank~$u$ and~$f$ in the opposite order unless they are both contained in the same series~$A_{i \to }$ for some
    $i \in [q]$.
    By our observation~($\dagger$), the same holds even if both~$u$ and~$f$ are contained in~$A_{i \to}$ for some
    $i \in [q]$, unless~$f$ is an edge incident to~$u$
    that connects~$u$ with some vertex in a later color class (i.e., in some~$U_j$ with~$j>i$).
    Similarly,~$\hat{v}_3$ and~$\hat{v}_4$ rank~$u$ and~$f$ in the opposite order unless they are both contained in the same series~$B_{\to i}$ for some
    $i \in [q]$.
    By our observation~($\ddagger$), the same holds even if both~$u$ and~$f$ are contained in~$B_{\to i}$ for some
    $i \in [q]$, unless~$f$ is an edge incident to~$u$ that connects~$u$ with some vertex in an earlier color class (i.e., in some~$U_j$ with~$j<i$).

    Taking into account the preferences of the four base votes (out of whom exactly three prefer edge candidates to vertex candidates), we can conclude that
    \begin{enumerate}[(i)]
        \item if~$f$ is \emph{not} incident to~$u$, then~$(u,f)$ is an arc of weight~$2$ in~$\majG$, and \label{prop:PP:Schulze:numParties:Wh:notIncident}
        \item if~$f$ is incident to~$u$, then neither~$(u,f)$ nor~$(f,u)$ is an arc in~$\majG$. \label{prop:PP:Schulze:numParties:Wh:incident}
    \end{enumerate}
    This completes the description of our construction. An overview of the final majority graph is shown in \Cref{fig:PP:Schulze:numParties:Wh:finalMajorityGraph}.

    \begin{figure}[bt!]
        \centering
        \begin{tikzpicture}[every node/.style={draw,circle,inner sep=0,minimum width=0.8cm}]
            \node (b) at (0,1) {$b$};

            \node (u11) at (-3.5,-1) {$u_i^1$};
            \node (u12) at (-2.5,-1) {$u_i^2$};
            \node (u13) at (-1.5,-1) {$u_i^3$};
            \draw[double,dotted,rounded corners] (-4.1,-0.4) rectangle (-0.9,-1.6);
            \node[draw=none,font=\footnotesize,myGray] at (-4.3,-0.7) {$U_i$};

            \node (u21) at (1.5,-1) {$u_j^1$};
            \node (u22) at (2.5,-1) {$u_j^1$};
            \node (u23) at (3.5,-1) {$u_j^3$};
            \draw[double,dotted,rounded corners] (0.9,-0.4) rectangle (4.1,-1.6);
            \node[draw=none,font=\footnotesize,myGray] at (4.35,-0.7) {$U_j$};

            \foreach \i in {1,2} {
                \foreach \j in {1,2,3} {
                    \draw[myarc,myGreen] (b) to (u\i\j.north);
                }
            }

            \def\ey{-3.5}
            \node (e121) at (-2.25,\ey) {$f_{i,j}^1$};
            \node (e122) at (-0.75,\ey) {$f_{i,j}^2$};
            \node (e123) at (0.75,\ey) {$f_{i,j}^3$};
            \node (e124) at (2.25,\ey) {$f_{i,j}^4$};
            \draw[double,dotted,rounded corners] (-2.85,-2.9) rectangle (2.85,-4.1);
            \node[draw=none,font=\footnotesize,myGray] at (3.2,-3.85) {$F_{i,j}$};

            \foreach \i in {11,12,22,23} {
                \draw[myarc,myRed] (u\i.south) to (e121.north);
            }
            \draw[dashed,myGray] (e121) edge (u13) edge (u21);
            \foreach \i in {12,13,21,23} {
                \draw[myarc,myRed] (u\i.south) to (e122.north);
            }
            \draw[dashed,myGray] (e122) edge (u11) edge (u22);
            \foreach \i in {11,13,21,22} {
                \draw[myarc,myRed] (u\i.south) to (e123.north);
            }
            \draw[dashed,myGray] (e123) edge (u12) edge (u23);
            \foreach \i in {11,12,21,23} {
                \draw[myarc,myRed] (u\i.south) to (e124.north);
            }
            \draw[dashed,myGray] (e124) edge (u13) edge (u22);

            \node (a) at (0,-5.5) {$a$};
            \foreach \i in {1,2,3,4} {
                \draw[myarc,myBlue] (e12\i) to (a);
            }

            \node (p) at (0,-7) {$p$};
            \draw[myarc,myRed] (a) to (p);
            \draw[myarc,double,myYellow,in=-90,out=120] (p) to (-2.5,-4.1);
            \draw[myarc,double,myYellow,in=-90,out=180] (p) to (-3.95,-1.6);
            \draw[myarc,double,myYellow,in=-90,out=0] (p) to (3.95,-1.6);
        \end{tikzpicture}
        \caption{An illustration of the final weighted majority graph~$\majG$ for the election~$\election$ constructed in the proof of \Cref{thm:PP:Schulze:numParties:Wh}. Here, we focus on exactly two color parties~$U_i$ and~$U_j$. The \textcolor{myYellow}{\bfseries yellow} arcs are due to new incidence votes, who all prefer~$p$ over any edge or vertex candidate. The coloring of the rest of the edges follows the same logic as in \Cref{fig:PP:Schulze:numParties:Wh:initialMajorityGraph}, since they are not affected by the incidence votes. Moreover, each dashed edge between a vertex candidate and an edge candidate represents that these two candidates are newly tied due to the incidence votes.}
        \label{fig:PP:Schulze:numParties:Wh:finalMajorityGraph}
    \end{figure}
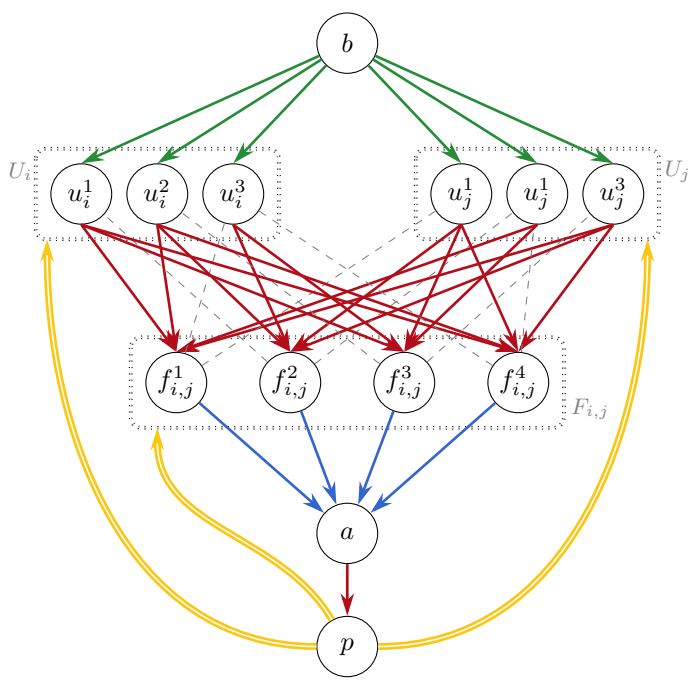

    \proofsubparagraph{Correctness.}
    First, assume that~$\mathcal{I}$ is a yes-instance and~$K$ is a multicolored clique in~$H$. We construct a solution for~$\mathcal{J}$ as follows. Each color party~$U_i$ nominates the vertex candidate corresponding to the vertex of color~$i$ in~$K$. Additionally, each edge party~$F_{i,j}$ nominates the edge candidate corresponding to the edge~$\{u_i^\ell,u_j^{\ell'}\} \in F$, where~$\{u_i^\ell,u_j^{\ell'}\}\subseteq K$. Observe that such an edge candidate always exists, since~$K$ induces a complete graph in~$K$. Now, consider the beatpaths in~$\majG$.
    Observe that there is no arc between any vertex candidate and any edge candidate due to property \eqref{prop:PP:Schulze:numParties:Wh:incident} and the fact that~$K$ is a clique. Since~$b$ is tied with both~$a$ and~$p$, we get that there is no path from~$b$ to any candidate in~$\{a,p\} \cup F$. Next, observe that 
    each edge candidate is on a cycle with candidates~$a$ and~$p$, and moreover,~$p$ has an arc of weight~$2$ to any vertex candidate. Therefore, there is a beatpath of strength~$2$ from~$p$ to each candidate other than~$b$.
    In fact, all the core candidates~$b$,~$a$,~$p$, and all edge candidates are winners in~$\election$. In particular,~$p$ is a possible president.

    In the opposite direction, assume that~$\nomination$ is a nomination under which candidate~$p$ is in the set of winners
    of the reduced election.
    We claim that there is no path of strength in~$\majG$ from~$b$ to~$p$: Indeed, assuming otherwise, we know that~$\beats(b,p) = 2$ but~$\beats(p,b) = 0$, 
    contradicting that~$p$ is a winner under nomination~$\nomination$.
    Therefore, we construct a vertex set~$K$ that, for every~$i\in[q]$, contains the vertex~$u_i^\ell$ corresponding to the vertex candidate nominated by party~$U_i$.
    Again, for the sake of contradiction, assume that~$K$ does not induce a complete graph in~$H$, that is, there is a pair of candidates~$u_i,u_j\in K$ such that~$\{u_i,u_j\}\not\in F$.
    Consequently, by \eqref{prop:PP:Schulze:numParties:Wh:notIncident}, there has to be a pair of a nominated vertex candidate~$u$ (either~$u_i$ or~$u_j$) and edge candidate~$f \in F_{i,j}\cap \nomination$ with an arc~$(u,f)$ of weight~$2$.
    This, however, yields  a path from~$b$ to~$p$, which is not possible.
    Therefore,~$K$ induces a complete graph in~$H$, showing that~$\mathcal{I}$ is a yes-instance.

    \smallskip
    To wrap up, the reduction can be clearly done in polynomial time. Moreover, the number of parties is~$q + \binom{q}{2} + 3 \in \Oh{q^2}$. That is, the presented reduction is indeed a parameterized reduction, completing the proof.
\end{proof}


\subsection{Necessary President}

Next, we turn to studying the complexity of \SchNP, again starting with its classical complexity.


\subsubsection{Intractability for a Constant Number of Voters and Small Parties}

\begin{theorem}
\label{thm:NP:Schulze:threeVoters:coNPh}
  \SchNP is \coNPc, even if no party contains more than three candidates and the number
  of voters is an odd constant at least three.
\end{theorem}

\begin{proof}
  Containment of \SchNP in~\coNP is clear~\cite{rot-woi:c:possible-and-necessary-president-problem-in-schulze-voting}, since
  given an instance of the complement of this problem, for each member of the distinguished party (one after the other), we can nondeterministically guess a valid nomination of the other parties, and for such a nomination guessed, we can verify in polynomial time whether the nominee of the distinguished party does not win the resulting reduced election.

  We prove \coNPhness
  of \SchNP by a reduction from \BalancedSat
  to the complement of \SchNP.
    Let~$\varphi=\bigwedge_{j \in[m]} \clause_j$ be the input 3-CNF formula defined over a set~$X=\{x_1,\dots,x_n\}$ of variables.
    We let~$\clauses = \{\clause_1,\dots,\clause_m\}$ denote the set of clauses in~$\varphi$.
    We will use the notation~$\ol{X}=\{\ol{x}_1, \dots, \ol{x}_n\}$.

    \proofsubparagraph{Construction.}
    Let us create an election~$\election$ with three voters.
    We let~$\{p\}$ be the distinguished party, and we introduce two further singleton parties formed by candidates~$a$ and~$b$.
    For each variable~$x_i \in X$, we create a party~$\{x_i,\overline{x}_i\}$ containing two \emph{variable candidates} corresponding to the positive and negative form of~$x_i$.
    We further create four \emph{literal occurrence candidates},~$x_i^1$,~$x_i^2$,~$\ol{x}_i^1$, and~$\ol{x}_i^2$, corresponding to the four occurrences of each variable~$x_i \in X$.
    These candidates are grouped together into parties according to the clauses in~$\varphi$: For each~$j \in [m]$, we create a party~$P_j=\{\ell_j^1,\ell_j^2,\ell_j^3\}$; note that this is indeed a partitioning of all literal occurrence candidates.
    This completes the definition of the candidate set and the party structure in~$\election$.
    Note that, as promised, each party has at most three candidates.
    The preferences of our three voters are as follows:
    \begin{align*}
    v_1 \colon&\quad b\  \ol{x}_1^1\  \ol{x}_1^2\  \cdots\
    \ol{x}_n^1\  \ol{x}_n^2\
    x_1^1\  x_1^2\ \cdots\  x_n^1\  x_n^2\
    a\  p\  x_1\  \cdots\  x_n\  \ol{x}_1\  \cdots\  \ol{x}_n;
    \\
    v_2 \colon&\quad p\
    x_1\  \ol{x}_1^1\  \ol{x}_1^2\  \cdots\
    x_n\  \ol{x}_n^1\  \ol{x}_n^2\
    \ol{x}_1\  x_1^1\  x_1^2\  \cdots\
    \ol{x}_n\  x_n^1\  x_n^2\  a\ b;
    \\
    v_3 \colon&\quad a\ b\
    \ol{x}_n\  x_n^2\  x_n^1\  \cdots\
    \ol{x}_1\  x_1^2\  x_1^1\
    x_n\  \ol{x}_n^2\  \ol{x}_n^1\  \cdots\
    x_1\  \ol{x}_n^2\  \ol{x}_n^1\  p.
    \end{align*}

    \begin{figure}[bt!]
        \centering
        \begin{tikzpicture}
            \node[inner sep=3pt] (p) at (0,0) {$p$};
            \node[inner sep=2pt] (xn) at (-0.6,-1) {$x_n$};
            \node[inner sep=2pt] at (-1.3,-1) {$\cdots$};
            \node[inner sep=2pt] (xi) at (-2,-1) {$x_i$};
            \node[inner sep=2pt] at (-2.7,-1) {$\cdots$};
            \node[inner sep=2pt] (x1) at (-3.4,-1) {$x_1$};

            \node[inner sep=2pt] (ox1) at (0.6,-1) {$\ol{x}_1$};
            \node[inner sep=2pt] at (1.3,-1) {$\cdots$};
            \node[inner sep=2pt] (oxi) at (2,-1) {$\ol{x}_i$};
            \node[inner sep=2pt] at (2.7,-1) {$\cdots$};
            \node[inner sep=2pt] (oxn) at (3.4,-1) {$\ol{x}_n$};

            \node[inner sep=2pt] (ol11) at (-3.9,-2) {$\ol{x}_1^1$};
            \node[inner sep=2pt] at (-3.2,-2) {$\cdots$};
            \node[inner sep=2pt] (oli1) at (-2.5,-2) {$\ol{x}_i^1$};
            \node[inner sep=2pt] (oli2) at (-1.5,-2) {$\ol{x}_i^2$};
            \node[inner sep=2pt] at (0,-2) {$\cdots$};
            \node[inner sep=2pt] (li1) at (2.5,-2) {$x_i^1$};
            \node[inner sep=2pt] (li2) at (1.5,-2) {$x_i^2$};
            \node[inner sep=2pt] at (3.2,-2) {$\cdots$};
            \node[inner sep=2pt] (ln2) at (3.9,-2) {$x_n^2$};

            \node[inner sep=3pt] (a) at (0,-3) {$a$};
            \node[inner sep=3pt] (b) at (0,-4) {$b$};

            \draw[rounded corners,dashed] (-4.5,0.8) rectangle (4.5,-2.5);

            \draw[rounded corners,double,dotted] (-4,0.5) rectangle (4,-1.4);

            \draw[myarc,double,myBlue,line width=1pt,in=-90,out=70] (a) to (0.5,-1.4);

            \draw[myarc,double,myBlue,line width=1pt,in=-90,out=0] (b) to (3,-2.5);

            \draw (p) circle[radius=8pt,inner sep=1pt];
            \draw (a) circle[radius=8pt,inner sep=1pt];
            \draw (b) circle[radius=8pt,inner sep=1pt];
            \draw (x1) circle[radius=8pt,inner sep=1pt];
            \draw (xi) circle[radius=8pt,inner sep=1pt];
            \draw (xn) circle[radius=8pt,inner sep=1pt];
            \draw (ox1) circle[radius=8pt,inner sep=1pt];
            \draw (oxi) circle[radius=8pt,inner sep=1pt];
            \draw (oxn) circle[radius=8pt,inner sep=1pt];
            \draw (ol11) circle[radius=8pt,inner sep=1pt];
            \draw (oli1) circle[radius=8pt,inner sep=1pt];
            \draw (oli2) circle[radius=8pt,inner sep=1pt];
            \draw (li1) circle[radius=8pt,inner sep=1pt];
            \draw (li2) circle[radius=8pt,inner sep=1pt];
            \draw (ln2) circle[radius=8pt,inner sep=1pt];

            \draw[myarc,myGreen,line width=1pt]    (p) to (x1);
            \draw[myarc,myGreen,line width=1pt]    (p) to (xi);
            \draw[myarc,myGreen,line width=1pt]    (p) to (xn);
            \draw[myarc,myGreen,line width=1pt]    (p) to (ox1);
            \draw[myarc,myGreen,line width=1pt]    (p) to (oxi);
            \draw[myarc,myGreen,line width=1pt]    (p) to (oxn);

            \draw[myarc,myRed,line width=1pt] (x1) to (ol11);
            \draw[myarc,myRed,line width=1pt] (xi) to (oli1);
            \draw[myarc,myRed,line width=1pt] (xi) to (oli2);
            \draw[myarc,myRed,line width=1pt] (oxi) to (li1);
            \draw[myarc,myRed,line width=1pt] (oxi) to (li2);
            \draw[myarc,myRed,line width=1pt] (oxn) to (ln2);

            \draw[myarc,myGreen,line width=1pt] (ol11) to (a);
            \draw[myarc,myGreen,line width=1pt] (oli1) to (a);
            \draw[myarc,myGreen,line width=1pt] (oli2) to (a);
            \draw[myarc,myGreen,line width=1pt] (li1) to (a);
            \draw[myarc,myGreen,line width=1pt] (li2) to (a);
            \draw[myarc,myGreen,line width=1pt] (ln2) to (a);

            \draw[myarc,myRed,line width=1pt] (a) to (b);

            \draw[fill=white] (p) circle[radius=8pt,inner sep=1pt];
            \draw[fill=white] (a) circle[radius=8pt,inner sep=1pt];
            \draw[fill=white] (b) circle[radius=8pt,inner sep=1pt];
            \draw[fill=white] (x1) circle[radius=8pt,inner sep=1pt];
            \draw[fill=white] (xi) circle[radius=8pt,inner sep=1pt];
            \draw[fill=white] (xn) circle[radius=8pt,inner sep=1pt];
            \draw[fill=white] (ox1) circle[radius=8pt,inner sep=1pt];
            \draw[fill=white] (oxi) circle[radius=8pt,inner sep=1pt];
            \draw[fill=white] (oxn) circle[radius=8pt,inner sep=1pt];
            \draw[fill=white] (ol11) circle[radius=8pt,inner sep=1pt];
            \draw[fill=white] (oli1) circle[radius=8pt,inner sep=1pt];
            \draw[fill=white] (oli2) circle[radius=8pt,inner sep=1pt];
            \draw[fill=white] (li1) circle[radius=8pt,inner sep=1pt];
            \draw[fill=white] (li2) circle[radius=8pt,inner sep=1pt];
            \draw[fill=white] (ln2) circle[radius=8pt,inner sep=1pt];

            \node[inner sep=3pt] (p) at (0,0) {$p$};
            \node[inner sep=2pt] (xn) at (-0.6,-1) {$x_n$};
            \node[inner sep=2pt] at (-1.3,-1) {$\cdots$};
            \node[inner sep=2pt] (xi) at (-2,-1) {$x_i$};
            \node[inner sep=2pt] at (-2.7,-1) {$\cdots$};
            \node[inner sep=2pt] (x1) at (-3.4,-1) {$x_1$};

            \node[inner sep=2pt] (ox1) at (0.6,-1) {$\ol{x}_1$};
            \node[inner sep=2pt] at (1.3,-1) {$\cdots$};
            \node[inner sep=2pt] (oxi) at (2,-1) {$\ol{x}_i$};
            \node[inner sep=2pt] at (2.7,-1) {$\cdots$};
            \node[inner sep=2pt] (oxn) at (3.4,-1) {$\ol{x}_n$};

            \node[inner sep=2pt] (ol11) at (-3.9,-2) {$\ol{x}_1^1$};
            \node[inner sep=2pt] at (-3.2,-2) {$\cdots$};
            \node[inner sep=2pt] (oli1) at (-2.5,-2) {$\ol{x}_i^1$};
            \node[inner sep=2pt] (oli2) at (-1.5,-2) {$\ol{x}_i^2$};
            \node[inner sep=2pt] at (0,-2) {$\cdots$};
            \node[inner sep=2pt] (li1) at (2.5,-2) {$x_i^1$};
            \node[inner sep=2pt] (li2) at (1.5,-2) {$x_i^2$};
            \node[inner sep=2pt] at (3.2,-2) {$\cdots$};
            \node[inner sep=2pt] (ln2) at (3.9,-2) {$x_n^2$};

            \node[inner sep=3pt] (a) at (0,-3) {$a$};
            \node[inner sep=3pt] (b) at (0,-4) {$b$};
        \end{tikzpicture}
        \caption{An illustration of the majority graph of the election~$\election$ constructed in the proof of \Cref{thm:NP:Schulze:threeVoters:coNPh}.
        Note that certain arcs are missing from the picture, for reasons of visibility.
        The double blue arc represents the set of arcs from~$b$ to every candidate within the dashed rectangle;
        notation regarding the color of the arcs is the same as in \Cref{fig:PP:Schulze:threeVoters:NPh:tournament_path}.}
        \label{fig:NP:Schulze:threeVoters:coNPh:construction}
    \end{figure}
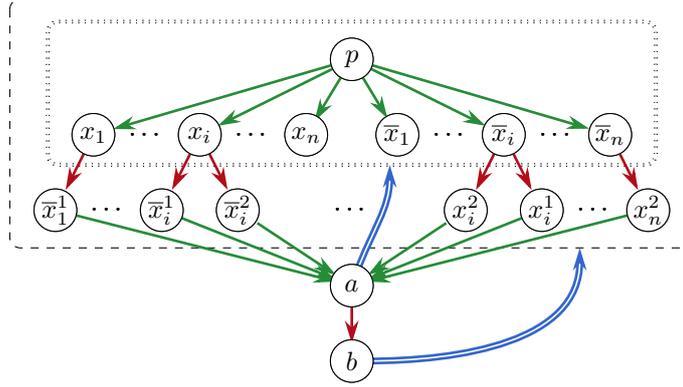

    An illustration of the majority graph of the constructed election~$\election$ is shown in \Cref{fig:NP:Schulze:threeVoters:coNPh:construction}.
    Note that all arcs in~$\majG$ have weight~$1$.

    \proofsubparagraph{Correctness.}
    We will use the following claim.

    \begin{claim}\label{clm:NP:Schulze:threeVoters:coNPh:necpress-path}
        Candidate~$p$ is a necessary president if and only if for every possible set~$\nomination$ of nominees containing exactly one candidate from each party, it holds that there is a path from~$p$ to~$a$ in~$\majGN$.
    \end{claim}
    \begin{claimproof}
        We observe the following facts:
        (i)~$b$ defeats all variable candidates, all literal occurrence candidates, and~$p$; and
        (ii)~$b$ is only defeated by~$a$.
        Now, if there were no path from~$p$ to~$a$ in~$\majGN$, then there would be no path in~$\majGN$ from~$p$ to~$b$ either, so by~$\weight(b,p)=1$ we have that~$p$ is not a winner in~$\election(N)$.
        Second, if there is a path from~$p$ to~$a$ in~$G$, then via~$b$ there is a path from~$p$ to every other candidate in~$N\setminus \{p\}$ as well.
    \end{claimproof}

    By \Cref{clm:NP:Schulze:threeVoters:coNPh:necpress-path}, the following property is enough to prove the correctness of the reduction.

    \begin{claim}\label{clm:NP:Schulze:threeVoters:coNPh:nec-path-sat}
        For every possible set~$\nomination$ of nominees containing exactly one candidate from each party,
        there exists a path from~$p$ to~$a$ in~$\majGN$ if and only if there exists \emph{no} satisfying truth assignment for~$\varphi$.
    \end{claim}
    \begin{claimproof}
        We are going to show the equivalence of the two statements by showing that their negated forms are equivalent.

        First, assume that~$\varphi$ admits a satisfying truth assignment~$\alpha$. We are going to show that there exists a set~$\nomination$ of nominees (containing a candidate from each party) for which there is no path from~$p$ to~$a$ in~$\majGN$.
        Let~$\nomination$ contain all three singletons (with $a$,~$b,$ and~$p$) as well as all variable candidates corresponding to true literals under~$\alpha$.
        Furthermore, for each clause~$\clause_j$, we select one literal occurrence candidate in~$P_j$ that corresponds to a true literal in~$\clause_j$ (such a candidate exists for each~$j \in [m]$, because~$\alpha$ satisfies~$\varphi$).

        We claim that there is no arc leaving the set
        \[
        R=N \cap (\{p\} \cup X \cup \ol{X})
        \]
        in~$\majGN$.
        First, both~$a$ and~$b$ have an arc pointing to every candidate in~$R$.
        Second, any literal occurrence candidate~$z^h$ defeats~$p$, and is defeated by exactly one variable candidate, namely, by~$\ol{z}$ (meaning~$\ol{x}_i$ if~$z^h \in \{x_i^1,x_i^2\}$, and~$x_i$ if~$z^h \in \{\ol{x}_i^1,\ol{x}_i^2\}$, since that is the only variable candidate that is preferred to~$z^h$ in both votes,~$v_2$ and~$v_3$).
        However, if~$z^h$ is a nominee in~$\nomination$, then it corresponds to a true literal, and thus,~$\ol{z}$ corresponds to a false literal and is therefore not in~$\nomination$.
        This shows that no candidate in~$N \setminus R$ is reachable from any candidate in~$R$, proving the first direction of our claim.

        Second, assume now that for a set~$\nomination$ of nominees (containing exactly one candidate from each party), there exists no path in~$\majGN$ from~$p$ to~$a$.
        We are going to construct a satisfying truth assignment~$\alpha$ for~$\varphi$ by simply setting those literals to \true\ for which the corresponding variable candidate is in~$\nomination$.
        It remains to show that this truth assignment satisfies each clause~$\clause_j$.
        To see this, observe that all nominated occurrence candidates~$z^h$ must correspond to a literal~$z$ with~$\alpha(z)=\true$: indeed, if~$\alpha(z)=\false$, then~$\ol{z} \in N$ and thus the path~$(p,\ol{z},z^h,a)$ is present in~$\majGN$; a contradiction.
        This proves that each clause has at least one true literal under~$\alpha$, namely the one that corresponds to its nominated literal occurrence candidate in~$P_j$.
    \end{claimproof}

    \Cref{clm:NP:Schulze:threeVoters:coNPh:necpress-path,clm:NP:Schulze:threeVoters:coNPh:nec-path-sat} show that the constructed instance of \SchNP is a yes-instance if and only if~$\varphi$ is not satisfiable. This proves the \coNPhness of the problem.
\end{proof}

\begin{theorem}\label{thm:NP:Schulze:fourVoters:coNPh}
  \SchNP is \coNPc, even if the maximum party size is two and the number
  of voters is an even constant at least four.
\end{theorem}
\begin{proof}
    Similarly as in  the proof of \Cref{thm:NP:Schulze:threeVoters:coNPh} we provide a reduction from \BalancedSat. Given a 3-CNF formula~$\varphi=\bigwedge_{j \in[m]} \clause_j$ with a set~$X=\{x_1,\dots,x_n\}$ of variables and a set~$\clauses = \{\clause_1,\dots,\clause_m\}$ of clauses, we will use the notation defined in the proof of \Cref{thm:NP:Schulze:threeVoters:coNPh}.

    \proofsubparagraph{Construction.}
    We let~$\{p\}$ be the distinguished party, and we also add three more candidates,~$a$,~$b$, and~$c$, each forming singleton parties.
    For each variable~$x_i \in X$, we create a party~$\{x_i,\overline{x}_i\}$ and   four \emph{literal occurrence candidates},~$x_i^1$,~$x_i^2$,~$\ol{x}_i^1$, and~$\ol{x}_i^2$.
    For each clause~$\clause_j$, we also introduce two \emph{selection candidates} in a party~$S_j=\{s_j^1, s_j^2\}$.
    We let~$S$ denote the set of all selection candidates.
    The literal occurrence candidates are grouped together into parties according to the clauses in~$\varphi$: For each~$j \in [m]$, there is a singleton party~$\{\ell_j^1\}$ and a another party~$P_j=\{\ell_j^2,\ell_j^3\}$.
    This completes the definition of the candidate set and the party structure. Observe that each party has at most two candidates.

    The preferences of our four voters are as follows:
    \begin{align*}
    v_1 \colon&\quad
    p\  a\  b\  c\  \ora{X \cup \ol{X}}\
    \ell_1^1\  s_1^1\  \ell_1^2\  \ell_1^3\  s_1^2\
    \ell_2^1\  s_2^1\  \ell_2^2\  \ell_2^3\  s_2^2\
    \cdots\
    \ell_m^1\  s_m^1\  \ell_m^2\  \ell_m^3\  s_m^2;
    \\
    v_2 \colon&\quad
    \ora{S}\  a\  b\  c\  p\
    x_1\  \ol{x}_1^1\  \ol{x}_1^2\  \cdots\
    x_n\  \ol{x}_n^1\  \ol{x}_n^2\
    \ol{x}_1\  x_1^1\  x_1^2\  \cdots\
    \ol{x}_n\  x_n^1\  x_n^2;
    \\
    v_3 \colon&\quad
    b\  c\
    \ell_m^2\  \ell_m^3\  s_m^2\  \ell_m^1\  s_m^1\
    \cdots\
    \ell_1^2\  \ell_1^3\  s_1^2\  \ell_1^1\  s_1^1\
    a\  p\  \ola{X \cup \ol{X}};
    \\
    v_4 \colon&\quad
    \ol{x}_n\  x_n^2\  x_n^1\  \cdots\
    \ol{x}_1\  x_1^2\  x_1^1\
    x_n\  \ol{x}_n^2\  \ol{x}_n^1\  \cdots\
    x_1\  \ol{x}_n^2\  \ol{x}_n^1\
    c\  \ola{S}\  a\  b\  p.
    \end{align*}
    The weighted majority graph is depicted in \Cref{fig:NP:Schulze:fourVoters:coNPh:construction}; the weight of each arc is~$2$.
    \begin{figure}[bt!]
        \centering
        \begin{tikzpicture}
            \node[inner sep=3pt] (p) at (0,0.2) {$p$};
            \node[inner sep=2pt] (xn) at (-0.6,-1) {$x_n$};
            \node[inner sep=2pt] at (-1.3,-1) {$\cdots$};
            \node[inner sep=2pt] (xi) at (-2,-1) {$x_i$};
            \node[inner sep=2pt] at (-2.7,-1) {$\cdots$};
            \node[inner sep=2pt] (x1) at (-3.4,-1) {$x_1$};
            \node[inner sep=2pt] (ox1) at (0.6,-1) {$\ol{x}_1$};
            \node[inner sep=2pt] at (1.3,-1) {$\cdots$};
            \node[inner sep=2pt] (oxi) at (2,-1) {$\ol{x}_i$};
            \node[inner sep=2pt] at (2.7,-1) {$\cdots$};
            \node[inner sep=2pt] (oxn) at (3.4,-1) {$\ol{x}_n$};
            \node[inner sep=2pt] (ol11) at (-3.9,-2) {$\ol{x}_1^1$};
            \node[inner sep=2pt] at (-3.2,-2) {$\cdots$};
            \node[inner sep=2pt] (oli1) at (-2.5,-2) {$\ol{x}_i^1$};
            \node[inner sep=2pt] (oli2) at (-1.5,-2) {$\ol{x}_i^2$};
            \node[inner sep=2pt] at (0,-2) {$\cdots$};
            \node[inner sep=2pt] (li1) at (1.5,-2) {$x_i^1$};
            \node[inner sep=1.5pt] (li2) at (2.5,-2) {$x_i^2$};
            \node[inner sep=2pt] at (3.2,-2) {$\cdots$};
            \node[inner sep=2pt] (ln2) at (3.9,-2) {$x_n^2$};

            \node[inner sep=1pt] (s11) at (-3.5,-3.3) {$s^1_1$};
            \node[inner sep=1pt] (s12) at (-2.5,-3.3) {$s^2_1$};
            \node at (-1.5,-3.3) {$\cdots$};
            \node[inner sep=1pt] (sj1) at (-0.5,-3.3) {$s^1_j$};
            \node[inner sep=1pt] (sj2) at (0.5,-3.3) {$s^2_j$};
            \node at (1.5,-3.3) {$\cdots$};
            \node[inner sep=1pt] (sm1) at (2.5,-3.3) {$s^1_m$};
            \node[inner sep=1pt] (sm2) at (3.5,-3.3) {$s^2_m$};

            \node[inner sep=3pt] (a) at (0,-4.6) {$a$};
            \node[inner sep=3pt] (b) at (0,-5.6) {$b$};
            \node[inner sep=4pt] (c) at (0,-6.6) {$c$};

            \draw (s11) circle[radius=8pt,inner sep=1pt];
            \draw (s12) circle[radius=8pt,inner sep=1pt];
            \draw (sj1) circle[radius=8pt,inner sep=1pt];
            \draw (sj2) circle[radius=8pt,inner sep=1pt];
            \draw (sm1) circle[radius=8pt,inner sep=1pt];
            \draw (sm2) circle[radius=8pt,inner sep=1pt];
            \draw (p) circle[radius=8pt,inner sep=1pt];
            \draw (a) circle[radius=8pt,inner sep=1pt];
            \draw (b) circle[radius=8pt,inner sep=1pt];
            \draw (c) circle[radius=8pt,inner sep=1pt];
            \draw (x1) circle[radius=8pt,inner sep=1pt];
            \draw (xi) circle[radius=8pt,inner sep=1pt];
            \draw (xn) circle[radius=8pt,inner sep=1pt];
            \draw (ox1) circle[radius=8pt,inner sep=1pt];
            \draw (oxi) circle[radius=8pt,inner sep=1pt];
            \draw (oxn) circle[radius=8pt,inner sep=1pt];
            \draw (ol11) circle[radius=8pt,inner sep=1pt];
            \draw (oli1) circle[radius=8pt,inner sep=1pt];
            \draw (oli2) circle[radius=8pt,inner sep=1pt];
            \draw (li1) circle[radius=8pt,inner sep=1pt];
            \draw (li2) circle[radius=8pt,inner sep=1pt];
            \draw (ln2) circle[radius=8pt,inner sep=1pt];

            \draw[myarc,myGreen,line width=1pt]    (p) to (x1);
            \draw[myarc,myGreen,line width=1pt]    (p) to (xi);
            \draw[myarc,myGreen,line width=1pt]    (p) to (xn);
            \draw[myarc,myGreen,line width=1pt]    (p) to (ox1);
            \draw[myarc,myGreen,line width=1pt]    (p) to (oxi);
            \draw[myarc,myGreen,line width=1pt]    (p) to (oxn);

            \draw[myarc,myRed,line width=1pt] (x1) to (ol11);
            \draw[myarc,myRed,line width=1pt] (xi) to (oli1);
            \draw[myarc,myRed,line width=1pt] (xi) to (oli2);
            \draw[myarc,myRed,line width=1pt] (oxi) to (li1);
            \draw[myarc,myRed,line width=1pt] (oxi) to (li2);
            \draw[myarc,myRed,line width=1pt] (oxn) to (ln2);

            \draw[myarc,myBlue,line width=1pt] (ol11) to (sj1);
            \draw[myarc,myBlue,line width=1pt] (li2) to (sj2);
            \draw[myarc,myBlue,line width=1pt] (ln2) to (sj2);

            \draw[myarc,myYellow,line width=1pt] (s11) to (a);
            \draw[myarc,myYellow,line width=1pt] (s12) to (a);
            \draw[myarc,myYellow,line width=1pt] (sj1) to (a);
            \draw[myarc,myYellow,line width=1pt] (sj2) to (a);
            \draw[myarc,myYellow,line width=1pt] (sm1) to (a);
            \draw[myarc,myYellow,line width=1pt] (sm2) to (a);

            \draw[myarc,myRed,line width=1pt] (a) to (b);
            \draw[myarc,myGreen,line width=1pt] (b) to (c);

            \draw[rounded corners,dashed] (-4.5,-0.3) rectangle (4.5,-2.5);
            \draw[rounded corners,double,dotted] (-4,-0.5) rectangle (4,-1.4);
            \draw[rounded corners,double,dotted] (-4.3,-2.85) rectangle (4.3,-3.75);

            \draw[myarc,double,myGreen,line width=1pt,in=-90,out=80] (a) to (-0.5,-1.4);
            \draw[myarc,double,myBlue,line width=1pt,in=-100,out=20] (c) to (3.3,-3.75);
            \draw[myarc,double,myGreen,line width=1pt,in=-90,out=0] (c) to (4,-2.5);
            \draw[myarc,double,myGreen,line width=1pt,in=-90,out=0] (b) to (3,-2.5);

            \draw[myYellow,line width=0.8pt, out=180, in=-90] (c) to (-5.4,-1.5);
            \draw[myarc,myYellow,line width=0.8pt, out=90, in=160] (-5.4,-1.5) to (p);

            \draw[myYellow,line width=0.8pt, out=180, in=-90] (b) to (-5.2,-1.5);
            \draw[myarc,myYellow,line width=0.8pt, out=90, in=170] (-5.2,-1.5) to (p);

            \draw[myYellow,line width=0.8pt, out=180, in=-90] (a) to (-5,-1.5);
            \draw[myarc,myYellow,line width=0.8pt, out=90, in=180] (-5,-1.5) to (p);

            \draw[fill=white] (s11) circle[radius=8pt,inner sep=1pt];
            \draw[fill=white] (s12) circle[radius=8pt,inner sep=1pt];
            \draw[fill=white] (sj1) circle[radius=8pt,inner sep=1pt];
            \draw[fill=white] (sj2) circle[radius=8pt,inner sep=1pt];
            \draw[fill=white] (sm1) circle[radius=8pt,inner sep=1pt];
            \draw[fill=white] (sm2) circle[radius=8pt,inner sep=1pt];
            \draw[fill=white] (p) circle[radius=8pt,inner sep=1pt];
            \draw[fill=white] (a) circle[radius=8pt,inner sep=1pt];
            \draw[fill=white] (b) circle[radius=8pt,inner sep=1pt];
            \draw[fill=white] (c) circle[radius=8pt,inner sep=1pt];
            \draw[fill=white] (x1) circle[radius=8pt,inner sep=1pt];
            \draw[fill=white] (xi) circle[radius=8pt,inner sep=1pt];
            \draw[fill=white] (xn) circle[radius=8pt,inner sep=1pt];
            \draw[fill=white] (ox1) circle[radius=8pt,inner sep=1pt];
            \draw[fill=white] (oxi) circle[radius=8pt,inner sep=1pt];
            \draw[fill=white] (oxn) circle[radius=8pt,inner sep=1pt];
            \draw[fill=white] (ol11) circle[radius=8pt,inner sep=1pt];
            \draw[fill=white] (oli1) circle[radius=8pt,inner sep=1pt];
            \draw[fill=white] (oli2) circle[radius=8pt,inner sep=1pt];
            \draw[fill=white] (li1) circle[radius=8pt,inner sep=1pt];
            \draw[fill=white] (li2) circle[radius=8pt,inner sep=1pt];
            \draw[fill=white] (ln2) circle[radius=8pt,inner sep=1pt];

            \node[inner sep=3pt] (p) at (0,0.2) {$p$};
            \node[inner sep=2pt] (xn) at (-0.6,-1) {$x_n$};
            \node[inner sep=2pt] at (-1.3,-1) {$\cdots$};
            \node[inner sep=2pt] (xi) at (-2,-1) {$x_i$};
            \node[inner sep=2pt] at (-2.7,-1) {$\cdots$};
            \node[inner sep=2pt] (x1) at (-3.4,-1) {$x_1$};

            \node[inner sep=2pt] (ox1) at (0.6,-1) {$\ol{x}_1$};
            \node[inner sep=2pt] at (1.3,-1) {$\cdots$};
            \node[inner sep=2pt] (oxi) at (2,-1) {$\ol{x}_i$};
            \node[inner sep=2pt] at (2.7,-1) {$\cdots$};
            \node[inner sep=2pt] (oxn) at (3.4,-1) {$\ol{x}_n$};

            \node[inner sep=2pt] (ol11) at (-3.9,-2) {$\ol{x}_1^1$};
            \node[inner sep=2pt] at (-3.2,-2) {$\cdots$};
            \node[inner sep=2pt] (oli1) at (-2.5,-2) {$\ol{x}_i^1$};
            \node[inner sep=2pt] (oli2) at (-1.5,-2) {$\ol{x}_i^2$};
            \node[inner sep=2pt] at (0,-2) {$\cdots$};
            \node[inner sep=2pt] (li1) at (1.5,-2) {$x_i^1$};
            \node[inner sep=2pt] (li2) at (2.5,-2) {$x_i^2$};
            \node[inner sep=2pt] at (3.2,-2) {$\cdots$};
            \node[inner sep=2pt] (ln2) at (3.9,-2) {$x_n^2$};

            \node[inner sep=3pt] (s11) at (-3.5,-3.3) {$s^1_1$};
            \node[inner sep=3pt] (s12) at (-2.5,-3.3) {$s^2_1$};
            \node[inner sep=3pt] (sj1) at (-0.5,-3.3) {$s^1_j$};
            \node[inner sep=3pt] (sj2) at (0.5,-3.3) {$s^2_j$};
            \node[inner sep=3pt] (sm1) at (2.5,-3.3) {$s^1_m$};
            \node[inner sep=3pt] (sm2) at (3.5,-3.3) {$s^2_m$};

            \node[inner sep=3pt] (a) at (0,-4.6) {$a$};
            \node[inner sep=3pt] (b) at (0,-5.6) {$b$};
            \node[inner sep=3pt] (c) at (0,-6.6) {$c$};
        \end{tikzpicture}
        \caption{The majority graph of the election~$\election$ constructed in the proof of \Cref{thm:NP:Schulze:fourVoters:coNPh}.
        Arcs between two candidates of the same type (e.g., when both are variable, literal occurrence, or  selection candidates) are missing; otherwise all arcs of the graph are shown.
        The picture assumes~$\clause_j=(\ol{x}_1 \vee x_i \vee x_n)$ with~$\ell_j^1=\ol{x}_1^1$, $\ell_j^2=x_i^2$, and~$\ell_j^3=x_n^2$.
        Notation regarding the color of the arcs is as follows: Arcs in \textcolor{myGreen}{\textbf{green}} represent preference of a candidate over another by votes~$v_1$,~$v_2,$ and~$v_3$;
        arcs in \textcolor{myRed}{\textbf{red}} by
        votes~$v_1$,~$v_2$, and~$v_4$;
        arcs in \textcolor{myBlue}{\textbf{blue}} by
        votes~$v_1$,~$v_3$, and~$v_4$; and finally,
        arcs in \textcolor{myYellow}{\textbf{yellow}} by
        votes~$v_2$,~$v_3,$ and~$v_4$.
        }
        \label{fig:NP:Schulze:fourVoters:coNPh:construction}
    \end{figure}
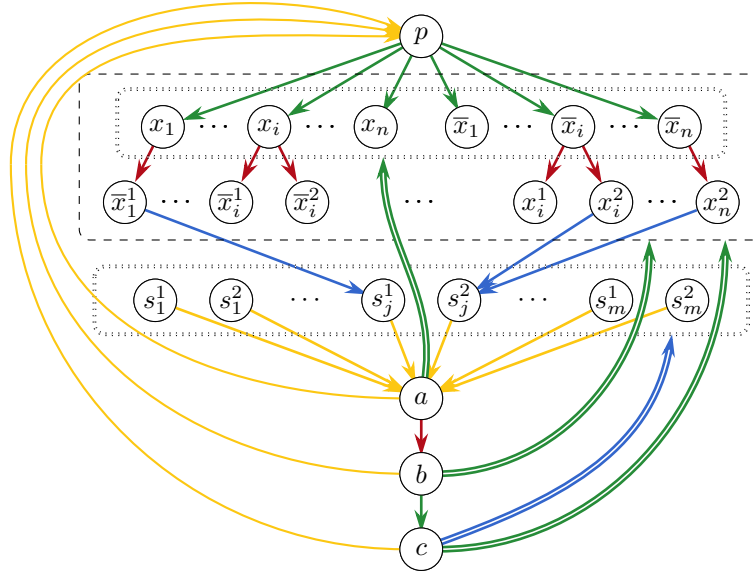

    \proofsubparagraph{Correctness.}
    Similarly as in the proof of \Cref{thm:NP:Schulze:threeVoters:coNPh}, we have the following claim.

    \begin{claim}\label{clm:NP:Schulze:fourVoters:coNPh:necpress-path}
        Candidate~$p$ is a necessary president if and only if for every possible set~$\nomination$ of nominees containing exactly one candidate from each party, it holds that there is a path from~$p$ to~$a$ in~$\majGN$.
    \end{claim}
    \begin{claimproof}
    Since~$a \in N$ and~$\weight(a,p)=2$, such a path is necessary for~$p$ to be a winner.
    Second, since~$b,c \in N$ and~$\weight(a,b)=\weight(b,c)=\weight(c,z)=2$ for each candidate~$z \notin \{a,b\}$, a path from~$p$ to~$a$ also yields a beatpath of strength~$2$ from~$a$ to every nominee.
    \end{claimproof}

    Next, we can prove the following:    \begin{claim}\label{clm:NP:Schulze:fourVoters:coNPh:nec-path-sat}
        For every possible set~$\nomination$ of nominees containing exactly one candidate from each party,
        there exists a path from~$p$ to~$a$ in~$\majGN$ if and only if there exists \emph{no} satisfying truth assignment for~$\varphi$.
    \end{claim}
    \begin{claimproof}
        We are going to show the equivalence of the two statements by showing that their negated forms are equivalent.

        First, assume that~$\varphi$ admits a satisfying truth assignment~$\alpha$. We are going to show that there exists a set~$\nomination$ of nominees (containing a candidate from each party) for which there is no path from~$p$ to~$a$ in~$\majGN$.
        Let~$\nomination$ contain all singletons ($a$,~$b$,~$c$,~$p$, and~$\ell_j^1$ for each~$j \in [m]$) as well as all variable candidates corresponding to true literals under~$\alpha$.
        Furthermore, for each clause~$\clause_j$, we fix a true literal~$\ell_j^\star$ occurring in~$\clause_j$ (which exists because~$\alpha$ satisfies~$\varphi$) and proceed as follows:
        If~$\ell_j^\star=\ell_j^1$, then we nominate the selection candidate~$s_j^1$ and an arbitrary candidate from party~$P_j=\{\ell_j^2,\ell_j^3\}$;
        if~$\ell_j^\star \in P_j$, then we nominate the selection candidate~$s_j^2$ and the literal occurrence candidate~$\ell_j^\star$.

        We claim that there is no arc entering the set
        \[T=\{a,b,c\} \cup S \cup \{\ell_j^\star:j \in [m]\}\]
        in~$\majGN$. Assume for the sake of contradiction that some arc~$e$ enters~$T$.
        We distinguish three cases:
        \begin{itemize}
            \item
        First, assume that~$e$ has some candidate~$s_j^1$ as its endpoint; then~$s_j^1 \in N$ implies that~$\ell_j^\star=\ell_j^1$, which in turn yields~$\ell_j^1 \in T$.
        Due to the preferences of votes~$v_1$ and~$v_3$, the only literal occurrence candidate that defeats~$s_j^1$ is~$\ell_j^1$, and so we must have~$e=(\ell_j^1,s_j^1)$. However,~$e$ then has both its endpoints in~$T$; a contradiction.
        \item
        Second, assume that~$e$ has some candidate~$s_j^2$ as its endpoint; then~$s_j^2 \in N$ implies that~$\ell_j^\star =\ell_j^h$ for some~$h \in \{2,3\}$, which in turn yields~$\ell_j^h \in T$.
        Again, due to the preferences of votes~$v_1$ and~$v_3$, the only literal occurrence candidates that defeat~$s_j^2$ are~$\ell_j^2$ and~$\ell_j^3$, forming party~$P_j$; therefore, we must have~$e=(\ell_j^h,s_j^2)$. However,~$e$ then has again both its endpoints in~$T$; a contradiction.
        \item
        Third, assume that~$e$ has some literal occurrence candidate~$\ell_j^\star$ as its endpoint.
        Considering the preferences of votes~$v_2$ and~$v_4$, the only candidate that defeats~$\ell_j^h$ is the variable candidate~$z$ that is the \emph{negated} form of the literal corresponding to~$\ell_j^h$, that is,~$z=\ol{x}_i$ if~$\ell_j^\star \in \{x_i^1,x_i^2\}$, and~$z=x_i$ if~$\ell_j^\star \in \{\ol{x}_i^1,\ol{x}_i^2\}$.
        However, since~$\ell_j^\star$ is a true literal under~$\alpha$,~$z$ is a false literal under~$\alpha$, and thus the set~$\nomination$ contains the variable candidate~$\ol{z}$ and not~$z$. Thus there exists no arc entering~$\ell_j^h$ in~$G(\election(N))$, contradicting our assumption on~$e$.
        \end{itemize}
        This proves that there is no path from~$p$ to~$a$ in~$G(\election(N))$, as promised.

        \smallskip
        Second, to show the converse implication, assume now that for a set~$\nomination$ of nominees (containing exactly one candidate from each party), there exists no path in~$\majGN$ from~$p$ to~$a$.
        We are going to construct a satisfying truth assignment~$\alpha$ for~$\varphi$ by simply setting those literals to \true\ for which the corresponding variable candidate is in~$\nomination$.
        Let us show that~$\alpha$ satisfies each clause~$\clause_j$.

        First, if~$s_j^1 \in N$, then we claim that~$\ell_j^1$ is a true literal.
        Let~$z$ denote the corresponding variable candidate, and assume for the sake of contradiction that~$z$ is a false literal, meaning that the literal candidate~$\ol{z}$ is a nominee in~$\nomination$.
        However, since~$\ol{z}$ defeats the literal occurrence candidates~$z^1$ and~$z^2$, by~$\ell_j^1 \in \{z^1,z^2\}$ we get that~$(\ol{z},\ell_j^1)$ is an arc in~$G(\election(N))$, giving rise to the path~$(p,\ol{z},\ell_j^1,s_j^1,a)$; a contradiction to our assumption on~$\nomination$.

        Second, if~$s_j^2 \in N$, then we claim that the unique literal candidate~$\ell_j^\star$ in~$P_j \cap N$ corresponds to a true literal.
        Let~$z$ denote the corresponding variable candidate, and assume for the sake of contradiction that~$z$ is a false literal, implying that~$\ol{z} \in N$.
        However,~$\ol{z}$ defeats both literal candidates,~$z^1$ and~$z^2$.
        Thus, as before, we get that~$(\ol{z},\ell_j^\star)$ is an arc in~$G(\election(N))$, giving rise to the path~$(p,\ol{z},\ell_j^\star,s_j^2,a)$; a contradiction to our assumption on~$\nomination$.
        This proves that each clause has at least one true literal under~$\alpha$.
    \end{claimproof}

    \Cref{clm:NP:Schulze:fourVoters:coNPh:necpress-path} and \Cref{clm:NP:Schulze:fourVoters:coNPh:nec-path-sat} together prove the \coNPhness of the problem.
\end{proof}


\subsubsection{Parameterized Hardness for Few Parties}

Finally, we study the parameterized complexity of \SchNP, again starting with an easy observation.

\begin{observation}\label{thm:NP:Schulze:numParties:XP}\label{thm:NP:Schulze:numParties:partySize:FPT}
    \SchNP is
    \begin{itemize}
        \item in \XP when parameterized by the number of parties, and
        \item fixed-parameter tractable when parameterized by the number of parties and the maximum party size, combined.
    \end{itemize}
\end{observation}

In the following, we show that the simple \XP algorithm for parameterization by the number of parties from the above observation is, under standard theoretical assumptions, essentially optimal.

\begin{theorem}\label{thm:NP:Schulze:numParties:coWh}
  When parameterized by the number~$\numParties$ of parties, it is \coWh to decide \SchNP, even if the number of voters
  is an odd constant at least seven.
\end{theorem}
\begin{proof}
  We again reduce from the \Wh \probName{Multicolored Clique} problem~\cite{fel-her-ros-via:j:multicolored-clique,pie:j:multicolored-clique},
  this time to the complement of \SchNP.
  Recall that in this problem, we are given a~$q$-partite graph~$H=(U = U_1 \cup \cdots U_q, F)$, and the goal is to decide whether a subset~$K\subseteq U$ of size~$q$ exists so that~$G[K]$ is a complete graph. The notation is the same as in \Cref{thm:PP:Schulze:numParties:Wh}, that is~$U_i$,~$i\in[q]$, is a \emph{color class} and~$F_{i,j}$ is the set of all edges between color classes~$U_i$ and~$U_j$.

    \proofsubparagraph{Construction.}
    Given an instance~$\mathcal{I} = (H,q)$, we construct an equivalent instance of \SchNP with election~$\election$ as follows.
    We have three \emph{core candidates}~$a$,~$b$, and~$p$; each of them forms a singleton party, and
    $\{p\}$ is the distinguished party.
    Next, for each vertex $u$ in~$U$, we introduce a \emph{vertex candidate}, also denoted as~$u$, and for each
    $i \in [q]$, we add the color class~$U_i$ as a \emph{color party}.
    Next, for each edge~$f \in F$, we introduce an \emph{edge candidate}, also denoted as~$f$, and for each
    $i,j \in [q]$ with~$i<j$, we add~$F_{i,j}$ as an \emph{edge party}.

    \begin{figure}[bt!]
        \centering
        \begin{tikzpicture}[
            every node/.style={draw,circle,inner sep=0,minimum width=0.8cm},
            v1v2/.style={myGreen},
            v2v3/.style={myRed},
            v1v3/.style={myBlue}
        ]
            \node (b) at (0,1) {$p$};

            \node (u11) at (-3.5,-1) {$u_i^1$};
            \node (u12) at (-2.5,-1) {$u_i^2$};
            \node (u13) at (-1.5,-1) {$u_i^3$};
            \draw[double,dotted,rounded corners] (-4.1,-0.4) rectangle (-0.9,-1.6);
            \node[draw=none,font=\footnotesize,myGray] at (-4.3,-0.7) {$U_i$};

            \node (u21) at (1.5,-1) {$u_j^1$};
            \node (u22) at (2.5,-1) {$u_j^1$};
            \node (u23) at (3.5,-1) {$u_j^3$};
            \draw[double,dotted,rounded corners] (0.9,-0.4) rectangle (4.1,-1.6);
            \node[draw=none,font=\footnotesize,myGray] at (4.35,-0.7) {$U_j$};

            \foreach \i in {1,2} {
                \foreach \j in {1,2,3} {
                    \draw[myarc,v1v2] (b) to (u\i\j.north);
                }
            }

            \def\ey{-3.5}
            \node (e121) at (-2.25,\ey) {$f_{i,j}^1$};
            \node (e122) at (-0.75,\ey) {$f_{i,j}^2$};
            \node (e123) at (0.75,\ey) {$f_{i,j}^3$};
            \node (e124) at (2.25,\ey) {$f_{i,j}^4$};
            \draw[double,dotted,rounded corners] (-2.85,-2.9) rectangle (2.85,-4.1);
            \node[draw=none,font=\footnotesize,myGray] at (3.2,-3.85) {$F_{i,j}$};

            \foreach \i in {11,12,22,23} {
                \draw[myarc,v2v3] (u\i.south) to (e121.north);
            }
            \draw[myarc,myYellow] (e121) edge (u13) edge (u21);
            \foreach \i in {12,13,21,23} {
                \draw[myarc,v2v3] (u\i.south) to (e122.north);
            }
            \draw[myarc,myYellow] (e122) edge (u11) edge (u22);
            \foreach \i in {11,13,21,22} {
                \draw[myarc,v2v3] (u\i.south) to (e123.north);
            }
            \draw[myarc,myYellow] (e123) edge (u12) edge (u23);
            \foreach \i in {11,12,21,23} {
                \draw[myarc,v2v3] (u\i.south) to (e124.north);
            }
            \draw[myarc,myYellow] (e124) edge (u13) edge (u22);

            \node (a) at (0,-5.5) {$a$};
            \foreach \i in {1,2,3,4} {
                \draw[myarc,v1v2] (e12\i) to (a);
            }
            \draw[line width=5pt,white,out=90,in=240] (a) to (0.85,-1);
            \draw[line width=5pt,white,out=90,in=-60] (a) to (-0.85,-1);
            \draw[myarc,double,v1v3,out=90,in=240] (a) to (0.85,-1);
            \draw[myarc,double,v1v3,out=90,in=-60] (a) to (-0.85,-1);

            \node (p) at (0,-7) {$b$};
            \draw[myarc,v2v3] (a) to (p);
            \draw[myarc,double,myYellow,in=-90,out=60] (p) to (2.5,-4.1);
            \draw[myarc,double,v1v3,in=-90,out=180] (p) to (-3.95,-1.6);
            \draw[myarc,double,v1v3,in=-90,out=0] (p) to (3.95,-1.6);
        \end{tikzpicture}
        \caption{An illustration of the final weighted majority graph~$\majG$ for the election~$\election$ constructed in the proof of \Cref{thm:NP:Schulze:numParties:coWh}. Here, we focus on exactly two color parties,~$U_i$ and~$U_j$.
          An arc~$(s,t)$ is colored in \textcolor{myGreen}{\textbf{green}} if an agent~$s$ is preferred in votes~$v_1$ and~$v_2$ (but not in~$v_3$); in \textcolor{myRed}{\textbf{red}} if~$s$ is preferred in~$v_2$ and~$v_3$ (but not in~$v_1$); and in \textcolor{myBlue}{\textbf{blue}} if~$s$ is preferred in~$v_1$ and~$v_3$ (but not in~$v_2$).
          \textcolor{myYellow}{\textbf{Yellow}} arcs are due to the incidence votes~$\hat{v}_1$,~$\hat{v}_2$,~$\hat{v}_3$, and~$\hat{v}_4$, and without these (i.e., if we focus just on the subprofile induced by the base votes), they go in the opposite direction.}
        \label{fig:NP:Schulze:numParties:Wh:construction}
    \end{figure}
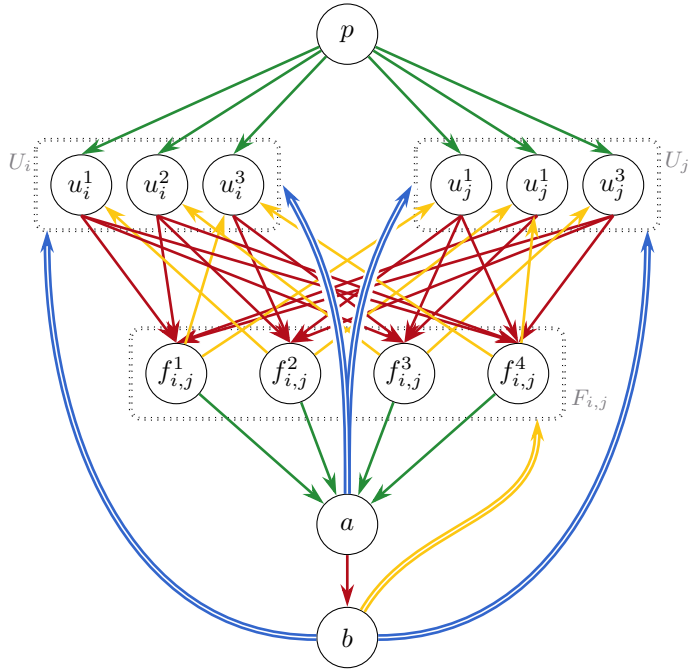

    Now, we define the votes. First, we have a group of three \emph{base votes},~$v_1$,~$v_2$, and~$v_3$. These votes create a base for the majority graph similarly to the weighted majority graph in the proof of
        \Cref{thm:PP:Schulze:numParties:Wh}.
    Their preferences are defined as follows:
    \begin{align*}
        v_1 \colon&\quad b\  \ora{F_{1,2} \cup \cdots \cup F_{q-1,q}}\  a\  p\  \ora{U_1\cup \cdots \cup U_q};\\
        v_2 \colon&\quad p\  \ora{U_1\cup \cdots \cup U_q}\  \ora{F_{1,2} \cup \cdots \cup F_{q-1,q}}\  a\  b;\\
        v_3 \colon&\quad a\  b\  \ola{U_1\cup \cdots \cup U_q}\  \ola{F_{1,2} \cup \cdots \cup F_{q-1,q}}\  p.
    \end{align*}

    Observe that the weighted majority graph of the election induced by these three votes is as follows:
    \begin{itemize}
        \item $b$ defeats all other candidates except for~$a$;
        \item $a$ beats~$b$,~$p$, and all vertex candidates;
        \item each edge candidate beats~$a$ and~$p$; and
        \item $p$ beats each vertex candidate.
    \end{itemize}
    Clearly, all arcs implied by these relations have weight one in the weighted majority graph.
    Also, observe that the distinguished candidate~$p$ is a winner for any nominations for this core subprofile. To avoid this situation, we add four additional \emph{incidence votes}.
    To define the preferences of these votes, we will use the notation~$A_{i\to}$,~$\wt{A}_{i\to}$,~$B_{\to i}$, and~$\wt{B}_{\to i}$ as defined in the proof of \Cref{thm:PP:Schulze:numParties:Wh}, and they still satisfy the properties ($\dagger$) and ($\ddagger$).

    We are now ready to give the preferences of our incidence votes:
    \begin{align*}
        \hat{v}_1 \colon& \quad
        p\  a\  b\  A_{1 \to}\  A_{2 \to}\ \cdots\  A_{(q-1) \to}\
        A_{q \to};
        \\
        \hat{v}_2 \colon& \quad
        \wt{A}_{q \to}\  \wt{A}_{(q-1) \to}\  \cdots\  \wt{A}_{2 \to}\  \wt{A}_{1 \to}\  b\  a\  p;
        \\
        \hat{v}_3 \colon& \quad
        p\  a\  b\
        B_{1 \to}\  B_{2 \to}\ \cdots\  B_{(q-1) \to}\
        B_{q \to};
        \\
        \hat{v}_4 \colon& \quad
        \wt{B}_{q \to}\  \wt{B}_{(q-1) \to}\  \cdots\  \wt{B}_{2 \to}\  \wt{B}_{1 \to}\  b\ a\ p.
    \end{align*}

    Observe that these votes do not change the relations between a core candidate~$x \in \{a,b,p\}$ and any other candidate~$y$, because~$\hat{v}_1$ and~$\hat{v}_2$, as well as~$\hat{v}_3$ and~$\hat{v}_4$, rank~$x$ and~$y$ in the opposite order.

    Let us now turn our attention to the relations between a vertex candidate~$u$ and some edge candidate~$f$.
    First, votes~$\hat{v}_1$ and~$\hat{v}_2$ clearly rank~$u$ and~$f$ in the opposite order unless they are both contained in the same series~$A_{i \to }$ for some
    $i \in [q]$. By our observation~($\dagger$), the same holds even if both~$u$ and~$f$ are contained in~$A_{i \to}$ for some
    $i \in [q]$, unless~$f$ is an edge incident to~$u$
    that connects~$u$ with some vertex in a later color class (i.e., in some~$U_j$ with~$j>i$).
    Similarly, votes~$\hat{v}_3$ and~$\hat{v}_4$ rank~$u$ and~$f$ in the opposite order unless they are both contained in the same series~$B_{\to i}$ for some
    $i \in [q]$.
    By our observation~($\ddagger$), the same holds even if both~$u$ and~$f$ are contained in~$B_{\to i}$ for some
    $i \in [q]$, unless~$f$ is an edge incident to~$u$ that connects~$u$ with some vertex in an earlier color class (i.e., in some~$U_j$ with~$j<i$).

    Taking into account the preferences of the base votes (out of whom two prefer vertex candidates to edge candidates, and one prefers edge candidates to vertex candidates), we can conclude that
    \begin{itemize}
        \item if~$f$ is \emph{not} incident to~$u$, then~$(u,f)$ is an edge of weight~$1$ in the majority graph, and
        \item if~$f$ is incident to~$u$, then~$(f,u)$ is an edge of weight~$1$ in the majority graph.
    \end{itemize}
    See \Cref{fig:NP:Schulze:numParties:Wh:construction} for an illustration of the final weighted majority graph.

    \proofsubparagraph{Correctness.}
    Assume first that~$\mathcal{I}$ is a yes-instance and~$K\subseteq U$ is a multicolored clique in~$H$.
    We create a set of nominations: First, each color party~$U_i$ nominates the vertex candidate corresponding to the vertex of~$U_i\cap K$.
    Additionally, each edge party~$F_{i,j}$ nominates the edge candidate corresponding to the edge connecting the vertices of~$U_i\cap K$ and~$U_j\cap K$ in~$H$.
    Note that such an edge candidate always exists, since~$K$ induces a complete graph.
    The resulting set~$\nomination$ of nominees, containing also all core candidates, has the following property: Candidate~$b$ defeats~$p$ in the reduced election~$\election(N)$ with~$\beats_N(b,p)=1$, but there is no path from~$p$ to~$b$ in the majority graph of~$\election(N)$.
    To see this, observe that no arc leaves the set~$\{p\} \cup \{U \cap N\}$, because (i)~$p$ defeats no candidates other than vertex candidates, and (ii) no vertex candidate defeats~$a$,~$b$, or any edge candidate, since every arc between the nominated vertex and edge candidates is directed from the latter to the former, as observed above.
    Hence, the distinguished candidate~$p$ is not a winner of the election~$\election(N)$.

    In the opposite direction, assume that there is a valid set~$\nomination$ of nominees such that~$p$ is \emph{not} a winner of the reduced election~$\election(N)$.
    We set~$K = U \cap N$, that is, we add to~$K$ the vertices corresponding to vertex candidates nominated by the color parties.
    As each color party has to nominate at least one candidate,~$K$ clearly contains one vertex from each color class, and therefore,~$|K| = q$.
    For the sake of contradiction, assume that~$K$ does not induce a complete
    subgraph in~$H$.
    Then, there exists a pair of vertices~$u_i^\ell$ and~$u_j^{\ell'}$ such that~$\{u_i^\ell,u_j^{\ell'}\} \not\in F$.
    Since~$\{u_i^\ell,u_j^{\ell'}\} \not\in F$, there is no edge candidate~$f \in F_{i,j}$ that beats both~$u_i^\ell$ and~$u_j^{\ell'}$.
    Therefore, the edge party~$F_{i,j}$ necessarily nominated an edge candidate, say~$e$, such that~$u_i^\ell$ or~$u_j^{\ell'}$ (or both) are preferred by the majority of voters over~$e$.
    Thus there is a path~$(p,(u_i^\ell \text{ or } u_j^{\ell'}), e, a, b)$ in the weighted majority graph.
    Since~$b$ defeats every candidate other than~$a$, this means that candidate~$p$ has a beatpath of strength~$1$ to every other nominee (while all paths leading to~$p$ have weight at most~$1$), and therefore~$p$ is a winner in~$\election(N)$.
    This contradicts that~$p$ is not a winner for nomination~$\nomination$.
    Hence, the edge~$\{u_i^\ell,u_j^{\ell'}\}$ has to exist in~$H$, implying that~$K$ is indeed a complete
    subgraph in~$H$.
    That is,~$\mathcal{I}$ is a yes-instance.

    Finally, observe that the number of parties is exactly~$q + \binom{q}{2} + 3 \in \mathcal{O}(q^2)$ and the reduction can be done in polynomial time.
    Hence, we indeed presented a parameterized reduction, completing the proof.
\end{proof}


\section{Conclusions and Future Research}

Complementing and improving the results of Rothe and Woitaschik~\cite{rot-woi:c:possible-and-necessary-president-problem-in-schulze-voting} on the \probName{Possible} and the \probName{Necessary President} problem, we have performed a more fine-grained parameterized complexity analysis for these two problems.
Our main results are dichotomies for them with respect to the number of voters: For two voters, they both are solvable in linear time, whereas we show the \probName{Possible President} problem \NPc for three or more voters and maximum party size two, and the \probName{Necessary President} problem \coNPc for three or more voters and maximum party size two or three, respectively, depending on whether the number of voters is even or odd.
In addition to these two natural parameters, we have also considered parameterization by the number of parties.
For any choice of viewing these three values as being a constant, a parameter, or unbounded, we have determined the parameterized complexity of the \probName{Possible} and the \probName{Necessary President} problem with respect to the corresponding parameterization.

In future research, we propose to study these two problems for other natural voting rules, especially those standing out because of their desirable axiomatic properties, such as ranked pairs~\cite{tid:j:independence-of-clones},\footnote{We note here that Cechlárová and Schlotter~\cite{cec-sch:c:necessary-president} study the complexity of the \probName{Necessary President} problem with respect to this rule, but, to the best of our knowledge, the complexity of \probName{Possible President} is still unknown.} where we have to pay particular attention the issue of breaking ties,\footnote{Note that winner determination in ranked pairs is \NPc for the ``parallel universe tie-breaking'' model~\cite{con-rog-xia:c:mle}, whereas it can be done efficiently when ranked pairs is defined by using a fixed tie-breaking method so that it becomes a resolute rule~\cite{bri-fis:c:price-of-neutrality}.} split cycle~\cite{hol-pac:j:split-cycle}, or the recently introduced River method~\cite{dor-bri-hei:c:river-method}.


\bibliography{PossiblePresident}

\end{document}